\documentclass[11pt,a4paper,reqno,intlimits,sumlimits]{amsart}
\usepackage{amssymb}
\usepackage{color}
\usepackage[mathscr]{euscript}
\usepackage{xspace}
\usepackage{enumerate}
\usepackage{epsfig,rotating}
\usepackage{graphicx}
\input{xy}\xyoption{all}
\usepackage{todonotes}
\usepackage{amsmath}
\usepackage{lscape}
\usepackage[latin1]{inputenc}
\usepackage{chicago}
\usepackage{accents}
\usepackage{dsfont}
\usepackage{comment}
\usepackage{mathrsfs}
\usepackage{float}
\usepackage{footmisc}
\usepackage{ragged2e}
\usepackage{multirow}
\usepackage{multicol}
\usepackage{todonotes}
\usepackage{centernot}

\usepackage[left=1.42in,right=1.42in,top=1.6in,bottom=1.3in]{geometry}

\DeclareMathAlphabet{\mathpzc}{OT1}{pzc}{m}{it}

\usepackage[bookmarksopen,pdfstartview=FitH]{hyperref}
\hypersetup{colorlinks,%
            citecolor=blue,%
            filecolor=blue,%
            linkcolor=blue,%
            urlcolor=blue}%

\theoremstyle{plain}
\newtheorem{theorem}{Theorem}[section]
\newtheorem{lemma}[theorem]{Lemma}
\newtheorem{proposition}[theorem]{Proposition}

\theoremstyle{definition}
\newtheorem{remark}[theorem]{Remark}

\newtheorem{definition}[theorem]{Definition}

\numberwithin{equation}{section}

\numberwithin{equation}{section}

\renewcommand{\setminus}{\backslash}

\makeatletter
\def\Ddots{\mathinner{\mkern1mu\raise\p@
\vbox{\kern7\p@\hbox{.}}\mkern2mu
\raise4\p@\hbox{.}\mkern2mu\raise7\p@\hbox{.}\mkern1mu}}
\makeatother

\newcommand{\de}[1]{\partial_{#1}}

\newcommand{\tr}[1]{\textup{Tr}\left[#1\right]}

\def\F{\ensuremath{\mathcal{F}}}

\def\R{\ensuremath{\mathbb{R}}}

\def\U{\ensuremath{\mathcal{U}}}
\def\Ec{\ensuremath{\mathcal{E}}}
\def\N{\mathbb{N}}

\def\B{\mathcal{B}}
\def\X{\mathcal{X}}

\def\M{\mathcal{M}}
\def\D{\mathcal{D}}
\def\G{\mathcal{G}}

\def\S{\mathcal{S}}

\newcommand{\diag}{\text{Diag}}
\newcommand{\SP}{\S_n^{+}}
\newcommand{\SPP}{\S_n^{+,*}}

\def\P{\mathbb{P}}
\def\Q{\mathbb{Q}}
\def\E{\mathbb{E}}

\def\hP{\ensuremath{\widehat{\mathrm{I\kern-.2em P}}}}
\def\hE{\ensuremath{\widehat{\mathrm{I\kern-.2em E}}}}

\def\bP{\ensuremath{\overline{\mathrm{I\kern-.2em P}}}}
\def\bE{\ensuremath{\overline{\mathrm{I\kern-.2em E}}}}

\def\y1{y_{-1}}

\newcommand{\q}[1]{\hspace*{-.7mm}\left\langle #1 \right\rangle}

\newcommand\smallO[1]{
    \,
    \mathchoice
    {
      {\scriptstyle\mathcal{O}}\left(#1\right)
    }
    {
      {\scriptstyle\mathcal{O}}\left(#1\right)
    }
    {
      {\scriptscriptstyle\mathcal{O}}\left(#1\right)
    }
    {
      \scalebox{0.8}{${\scriptscriptstyle\mathcal{O}}$}\left(#1\right)
    }
}
\makeatletter
\allowdisplaybreaks

\begin{document}
\begin{center}
{\Large \bf Long-time large deviations for the multi-asset Wishart stochastic volatility model and option pricing}
\vspace*{.5cm}
\begin{multicols}{3}
  Aur\'elien {\sc Alfonsi}\footnote{Universit\'e Paris-Est, Cermics (ENPC), INRIA, F-77455 Marne-la-Vall\'ee, France. 
  } \\
David {\sc Krief} \footnote{\label{fn:Paris7}LPSM, Universit\'e Paris Diderot, Paris, France}\\
Peter {\sc Tankov} \footnote{ENSAE ParisTech, Palaiseau, France}
\end{multicols} 
\vspace*{.5cm}
\end{center}
{\justifying {\it In this paper, we prove a large deviations principle
    for the class of multidimensional affine stochastic volatility
    models considered in (Gourieroux, C.~and Sufana, R.,
    J.~Bus.~Econ.~Stat., 28(3), 2010), where the volatility matrix
    is modelled by a Wishart process. This class extends the very
    popular Heston model to the multivariate setting,
    thus allowing to model the joint behaviour of a basket of stocks
    or several interest rates. We then use the large deviation
    principle to obtain an asymptotic approximation for the implied
    volatility of basket options and to develop an asymptotically optimal
    importance sampling algorithm, to reduce the number of simulations when using
    Monte-Carlo methods to price derivatives.}}

\medskip


\medskip


\setlength{\parindent}{0in}

\section{Introduction}

The Heston stochastic volatility model~\cite{Hes1993} is one of the
most popular models in quantitative finance for the evolution of a
single asset price. The Wishart stochastic volatility model is its natural extension to a basket of assets, since it coincides with the Heston model in dimension~$1$ and preserves the affine structure. This model, proposed in~\cite{Go:Su2004}, assumes that under the risk-neutral
probability, the vector of $n$ asset prices is modelled as an It\^o process
\begin{equation}\label{eq:S_tDaF}
dS_t = \diag(S_t) \, \left( r \mathbf{1} \,dt + \tilde{X}_t^{1/2} \, d\tilde{Z}_t \right) \:,
\end{equation}
where the $n\times n$ volatility matrix $(\tilde X_t)$ follows the Wishart process with dynamics
\begin{equation}\label{eq:X_tDaF}
d\tilde{X}_t = \left( \alpha \, a^{\top} a + \tilde{b} \tilde{X}_t + \tilde{X}_t \tilde{b}^{\top} \right) \,dt + \tilde{X}_t^{1/2} \, d\tilde{W}_t \, a + a^{\top} (d\tilde{W}_t)^\top \tilde{X}_t^{1/2} \:,
\end{equation}
where $\tilde{Z}$ and $\tilde{W}$ are independent standard
$n$-dimensional and $n\times n$-dimensional Brownian motions, and
$\diag(S_t)$ is the diagonal matrix whose diagonal elements are given
by the vector $S_t\in\R^n$. 

The
matrix process \eqref{eq:X_tDaF} has been introduced by \cite{Bru1991}
to model the perturbation of experimental biological data. As shown by
\cite{Bru1991} and \cite{Cu:Fi:Ma:Te2011} in a more general framework,
for $\alpha \ge n+1$ (resp. $\alpha \ge n-1$), the SDE
\eqref{eq:X_tDaF} has a unique strong (resp. weak)
solution. Furthermore, since $\tilde{X}_t$ is positive semi-definite
\cite[Prop.~4]{Bru1991}, Wishart processes turn out to be very
suitable processes to model covariance matrices. This, and the affine property of the Wishart process, led several
authors to use them in stochastic volatility models for a single
asset, such as \cite{DaF:Gr:Te2008} and \cite{Be:Be:ElK2008} and in
the Wishart stochastic volatility model  for multiple assets
\eqref{eq:S_tDaF}--\eqref{eq:X_tDaF}. Subsequently, this model has been extended by~\cite{DaF:Gr:Te2007} to include
a constant correlation between $W$ and $Z$ in a way to preserve the
affine structure.

By using the affine property, the Laplace transform of the model
\eqref{eq:S_tDaF}--\eqref{eq:X_tDaF}  is computed as follows~\cite{DaF:Gr:Te2007}. 
\begin{equation}\label{Laplace_intro}
\E\left(e^{\theta^\top \log(S_t)}\right) = \exp\left(\beta_\theta(t) + \tr{\gamma_\theta(t) \, \tilde{X}_0} + \delta_\theta^\top(t) \log(S_t) \right) \:,
\end{equation}
where $\beta_\theta, \gamma_\theta$ and $\delta_\theta$ satisfy the matrix Riccati equations
\begin{align*}
\de{t} \beta_\theta(t) & = r \, \delta_\theta^\top(t) \, \mathbf{1} + \alpha \, \tr{\gamma_\theta(t)} \\
\de{t} \gamma_\theta(t) & = \tilde{b}^\top \gamma_\theta(t) + \gamma_\theta(t) \, \tilde{b} + 2 \gamma_\theta(t) \, a^\top a \, \gamma_\theta(t) - \frac{1}{2}\,  \left(\diag(\delta_\theta(t)) - \delta_\theta(t) \delta_\theta^\top(t) \right)  \\
\de{t} \delta_\theta(t) & = 0 \:,
\end{align*}
with initial conditions $\beta_\theta(0) = 0$, $\gamma_\theta(0)=0$
and $\delta_\theta(0) = \theta$. Since the Riccati equations can be
solved explicitly, the Laplace transform can be expressed explicitly
in terms of matrix exponentials and inverses.

The goal of the present paper is to prove a large deviations principle
the Wishart stochastic volatility
model~\eqref{eq:S_tDaF}--\eqref{eq:X_tDaF} in the large-time
asymptotic regime. Since the Laplace transform of the log-price vector in the Wishart
model is known explicitly, a natural path towards a large deviations
principle is via G\"artner-Ellis theorem. However, despite the
explicit form of the Laplace transform, it is not easy to calculate
its long-time asymptotics and to check the assumptions of the theorem
because of the multi-dimensional setting. In this paper we therefore
focus on a (large enough) subclass of the model
\eqref{eq:S_tDaF}--\eqref{eq:X_tDaF} which enables us to obtain a
simpler formula for the limiting Laplace transform and then prove a large deviations principle. 

 Beyond its theoretical interest, knowing that a
given model satisfies a large deviations principle, and knowing the
explicit form of the rate function, enables one to develop a number of
important applications. One can mention e.g., efficient importance
sampling methods for Monte Carlo option pricing; asymptotic formulas for option prices
and implied
volatilities in various asymptotic regimes, approximate evaluation of risk
measures, simulation of rare events and others. We refer the reader to
\cite{pham2007some} for a review of various applications of large
deviations methods in finance. In this paper we develop applications
to variance reduction of Monte Carlo methods
and to the asymtotic computation of implied volatilities far from
maturity.  

Our variance reduction method follows previous works of
\cite{Gu:Ro2008}, \cite{Rob2010} and \cite{Ge:Ta} and uses Varadhan's lemma of large deviations theory to approximate the
optimal measure change in the importance sampling algorithm. 
Note that since the Laplace
  tranform is known explicitly, Fourier inversion methods can be used,
  as explained in \cite{DaF:Gr:Te2007}. However, these methods are
  much less competitive than in dimension~$1$ since they
  require to approximate an integral on $\R^n$. When, for complexity
  reasons, Fourier methods are not an option, the use of a large
  number of Monte-Carlo simulations is necessary. \cite{AA2013} present an exact simulation method for Wishart processes and a
  second order scheme for the Gourieroux and Sufana
  model~\eqref{eq:S_tDaF}--\eqref{eq:X_tDaF}. Thus, it is possible
  to sample efficiently such processes, and it is relevant to develop
  variance reduction techniques to reduce computational costs.

The approximation of implied volatility far from maturity extends
earlier results on the Heston model and the one-dimensional affine
stochastic volatility models \cite{FordeJacquier2011,Ja:Ke:Mi2013} to
the multidimensional setting of Wishart model. Once again, this
approach is more relevant in the multidimensional setting, since in
one-dimensional affine models the implied volatility may be quickly
computed by Fourier inversion.

In this paper, we denote $\M_n$ the set of real squared $n\times n$
matrices, $\S_n \subset \M_n$ the set of symmetric matrices and $\SP$,
(resp. $\SPP$), the sets of symmetric an non-negative (resp.) positive
definite. For a Borel set $A$, we denote by $\bar A$ the closure of
$A$ and by $\circ A$ the interior of $A$. 

The paper is structured as follows.
In Section \ref{sec:Model}, we describe the model, make certain assumptions on the parameters and give some properties of the model. In Section \ref{sec:LDP}, we prove that the asset log-price vector satisfies large deviations principle when maturity goes to infinity. In Section \ref{sec:AsymptoticPricing}, we calculate the asymptotic put basket implied volatility, following the approach of \cite{Ja:Ke:Mi2013}. In Section \ref{sec:VarianceReduction}, we develop the variance reduction method using Varadhan's lemma. Finally, in Section \ref{sec:Numerics}, we test numerically the results of Sections \ref{sec:AsymptoticPricing} and \ref{sec:VarianceReduction}.

\section{The Wishart stochastic volatility model}\label{sec:Model}
In this section we introduce the subclass of the Wishart stochastic
volatility models, in which we are interested in the present paper,
and compute the Laplace transform of the log stock price process. 

Let $(S_t)_{t \geq 0}$ be a $n$-dimensional vector stochastic process with dynamics 
\begin{equation}\label{eq:S_t}
dS_t = \diag(S_t) \, \left( r \mathbf{1} \,dt + a^{\top} X_t^{1/2} \,
  dZ_t \right) \:, \qquad S^i_0 > 0, \ i=1,\dots,n,
\end{equation}
where $\mathbf{1} = (1,...,1)^\top$, $\diag(S_t)_{ij}=  \mathds{1}_{\{i=j\}}S^i_t$, $Z_t$ is $n$-dimensional standard Brownian motion and the stochastic volatility matrix $X$ is a Wishart process with dynamics
\begin{equation}\label{eq:X_t}
dX_t = \left( \alpha I_n + b X_t + X_t b \right) \,dt + X_t^{1/2} \, dW_t + (dW_t)^\top X_t^{1/2} \:, \qquad X_0 = x \:.
\end{equation}
with $\alpha > n-1$, $a \in \M_n$ invertible, $-b,x \in \SPP$ and $W$
is a $n\times n$ matrix standard Brownian motion independent of $Z$. Note again that $X_t \in \SP$ \cite[Prop. 4]{Bru1991}. Let us also assume that $a$ is such that $a^\top a \in \SPP$. 
\begin{remark}
The model $(S,X)$ defined in \eqref{eq:S_t} and \eqref{eq:X_t} is a
(quite large) subclass of the one defined in \eqref{eq:S_tDaF} and \eqref{eq:X_tDaF}. Indeed, defining $\tilde{X}_t := a^\top X_t \, a$, we have $a^{\top} X_t^{1/2} \, dZ_t = \tilde{X}_t^{1/2} \, d\tilde{Z}_t$, where $\tilde{Z}_t$ is another $n$-dimensional standard Brownian motion and 
\[
d\tilde{X}_t = \left( \alpha \, a^{\top} a + \tilde{b} \tilde{X}_t + \tilde{X}_t \tilde{b}^{\top} \right) \,dt + \tilde{X}_t^{1/2} \, d\tilde{W}_t \, a + a^{\top} (d\tilde{W}_t)^\top \tilde{X}_t^{1/2} \:, \qquad \tilde{X}_0 = a^\top x \, a \:,
\]
where $\tilde{b} = a^{\top} \! b \, (a^{\top})^{-1}$ and $\tilde{W}_t$ is another $n\times n$-Brownian motion.
\end{remark}
\begin{remark}
In dimension one, the model defined by eqs. \eqref{eq:S_t} and \eqref{eq:X_t} corresponds to the famous Heston model \cite{Hes1993} and $b$ being negative definite yields the mean reversion property of the stochastic volatility process. 
\end{remark}
Defining the log-price $Y_t^k := \log(S_t^k), \, k=1,...,n$, a simple application of It\=o's lemma gives 
\begin{equation}\label{eq:Y_t}
d Y_t = \left( r \mathbf{1} - \frac{1}{2} \left( (a^\top X_t \, a)_{11}\,,\,...\,,\,(a^\top X_t \, a)_{nn}\right)^\top \right) \, dt + a^\top X_t^{1/2} \, dZ_t \:.
\end{equation}
We are interested in the Laplace transform of $Y_t$. In order to calculate it, we first cite the following proposition.
\begin{proposition}\textup{\cite[Prop. 5.1.]{Al:Ke:Re2016}\textbf{.}} \label{prop:Alfonsi}
Let $\alpha \ge n-1$, $x \in \SP$, $b \in \S_n$ and $X$ with dynamics \eqref{eq:X_t}. Let $v,w \in \S_n$ be such that 
\[
\exists m \in \S_n, \quad \frac{v}{2} - mb - bm - 2m^2 \in \SP \quad \text{and} \quad \frac{w}{2} + m \in \SP \:.
\] 
If $R_t := \int_0^t X_s \, ds$, then we have for $t \ge 0$
\begin{align*}
& \E\left[ \exp\left(- \frac{1}{2} \tr{w X_t} - \frac{1}{2} \tr{v R_t}\right) \right] \\
& \qquad = \frac{\exp\left(-\frac{\alpha}{2} \tr{b}t\right)}{\det\left[ V_{v,w}(t)\right]^{\alpha/2}} \, \exp\left( -\frac{1}{2} \tr{\left(V_{v,w}'(t) V_{v,w}^{-1}(t) + b\right) x}\right) \:,
\end{align*}
with
\[
V_{v,w}(t) = \left(\sum_{k=0}^\infty t^{2k+1} \frac{\tilde{v}^k}{(2k+1)!}\right) \tilde{w} + \sum_{k=0}^\infty t^{2k} \frac{\tilde{v}^{k}}{(2k)!}, \quad \tilde{v} = v+b^2 \quad \text{and} \quad \tilde{w} = w - b \:. 
\]
If besides, $\tilde{v} \in \SPP$, then
\[
V_{v,w}(t) = \tilde{v}^{-1/2} \sinh\left(\tilde{v}^{1/2} t\right) \tilde{w} + \cosh\left(\tilde{v}^{1/2} t\right)
\]
and
\[
V_{v,w}'(t) = \cosh\left(\tilde{v}^{1/2} t\right) \tilde{w} + \sinh\left(\tilde{v}^{1/2} t\right) \, \tilde{v}^{1/2}\:.
\]
\end{proposition}
The following proposition provides and explicit formula for the
Laplace transform of the log stock price $Y_t$ in the model \eqref{eq:S_t}--\eqref{eq:X_t}.
\begin{proposition}\label{prop:LaplaceExplicitForm}
Let $\phi: \R^n \rightarrow \S_n$ be the function defined by
\begin{equation}\label{eq:PhiDefinition}
\phi(\theta) := b^2 + a \, \left(\diag(\theta) - \theta\theta^\top \right) a^\top \in \S_n \:,
\end{equation}
Let $\U \subset \R^n$, be the set defined by
\[
\U := \left\{ \theta \in \R^n \,:\, \phi(\theta) \in \SP \right\}.
\]
Then, for all $\theta \in \U$,  the Laplace transform of $Y_t$  is 
\[
\E\left( e^{\theta^\top Y_t} \right) 
= \frac{e^{\theta^\top Y_0 + r \theta^\top \mathbf{1} \, t -\frac{\alpha}{2} \tr{b} t - \frac{1}{2} \tr{\left(b + \phi^{1/2}(\theta)\right) x - \exp\left(-t \, \phi^{1/2}(\theta)\right) \left( b + \phi^{1/2}(\theta)\right) V^{-1}(t)\, x}}}{\det\left[ V(t)\right]^{\alpha/2}}  \:,
\]
where 
\[
V(t) = \cosh\left(t \, \phi^{1/2}(\theta)\right) - \phi^{-1/2}(\theta) \, \sinh\left(t \, \phi^{1/2}(\theta)\right) \, b \:.
\]
\end{proposition}
\begin{proof}
By conditioning on the trajectory of $X$, we have
\[
\E\left( e^{\theta^\top Y_t} \right) = \E\left(\E\left( e^{\theta^\top Y_t} \,\middle|\, (X_s)_{s \le t}\right)\right)\:,
\]
where 
\begin{align*}
\E\left( e^{\theta^\top Y_t} \,\middle|\, (X_s)_{s \le t}\right)
& = e^{\theta^\top Y_0 + r \theta^\top \mathbf{1} \, t - \frac{1}{2} \int_0^t \theta^\top \left( (a^\top X_s \, a)_{11}\,,\,...\,,\,(a^\top X_s \, a)_{nn}\right)^\top \!\! - \theta^\top a^\top X_s \, a \, \theta\, ds} \\
& = e^{\theta^\top Y_0 + r \theta^\top \mathbf{1} \, t - \frac{1}{2} \int_0^t \tr{\diag(\theta) \, a^\top X_s \, a} - \tr{\theta^\top a^\top X_s \, a \, \theta}\, ds} \\
& = e^{\theta^\top Y_0 + r \theta^\top \mathbf{1} \, t - \frac{1}{2} \tr{a \, \left(\diag(\theta) - \theta\theta^\top \right) a^\top R_t }} \:.
\end{align*}
Let $m = -b/2$. Then $m \in \SP$ and
\[
\frac{a \, \left(\diag(\theta) - \theta\theta^\top \right) a^\top}{2} - mb - bm - 2m^2 = \frac{\phi(\theta)}{2} \in \SP \:.
\] 
Therefore, by Proposition \ref{prop:Alfonsi},
\begin{align}
\E\left( e^{\theta^\top Y_t} \right)
& = e^{\theta^\top Y_0 + r \theta^\top \mathbf{1} \, t} \, \E\left( e^{ - \frac{1}{2} \tr{a \, \left(\diag(\theta) - \theta\theta^\top \right) a^\top R_t }} \right) \notag \\
& = e^{\theta^\top Y_0 + r \theta^\top \mathbf{1} \, t} \, \frac{\exp\left(-\frac{\alpha}{2} \tr{b} t\right)}{\det\left[ V(t)\right]^{\alpha/2}} \, \exp\left( -\frac{1}{2} \tr{\left(V'(t) V^{-1}(t) + b\right) x}\right) \label{eq:LaplaceEquation}
\end{align}
where
\[
\begin{cases}
V(t) & = \cosh\left(t \, \phi^{1/2}(\theta)\right) - \phi^{-1/2}(\theta) \, \sinh\left(t \, \phi^{1/2}(\theta)\right) \, b \:, \\
V'(t) & = \sinh\left(t \, \phi^{1/2}(\theta)\right) \, \phi^{1/2}(\theta) - \cosh\left(t \, \phi^{1/2}(\theta)\right) \, b \:.
\end{cases}
\]
Since $\phi(\theta) \in \SP$, we can write $\phi(\theta) = P D P^\top$, where $D$ is diagonal, $P$ is orthonormal and $\hat{b} = - P^\top b \, P \in \SPP$.
\[
\begin{cases}
V(t) & = P \left( \cosh\left(t \, D^{1/2}\right) + \sinh\left(t \, D^{1/2}\right) D^{-1/2} \, \hat{b}\right) P^\top \:, \\
V'(t) & = P \left( \sinh\left(t \, D^{1/2}\right) D^{1/2} + \cosh\left(t \, D^{1/2}\right) \, \hat{b} \right) P^\top \\
& = \phi^{1/2}(\theta) V(t) - \exp\left(-t \, \phi^{1/2}(\theta)\right) \left( b + \phi^{1/2}(\theta)\right) \:.
\end{cases}
\]
Replacing $V'$ by the latter expression finishes the proof.\end{proof}
\begin{remark}
Note that, when $\phi(\theta) \in \SP \backslash \SPP$, $\phi^{1/2}(\theta)$ is not invertible. The notation $\phi^{-1/2}(\theta) \, \sinh\left(t \, \phi^{1/2}(\theta)\right)$ is therefore abusive and is to be interpreted as the finite limit 
\[
\lim_{\SPP \ni \, \phi \, \rightarrow \, \phi(\theta)} \phi^{-1/2} \, \sinh\left(t \, \phi^{1/2}\right)=\sum_{k=0}^\infty \frac{\phi(\theta)^k t^{2k+1}}{(2k+1)!} \:.
\]
\end{remark}

\begin{remark}\label{rem:bounded}
The set $\mathcal U$ is bounded. Indeed, let $\theta =
\lambda \bar\theta$, with $\lambda>0$ and $\|\bar \theta\|=1$. Then, letting $u =
(a^\top)^{-1}\bar \theta$, we have
$$
u^\top \phi(\theta) u = \|b (a^\top)^{-1}\bar \theta\|^2 + \lambda \bar
\theta^\top \text{Diag}(\bar\theta) \bar \theta - \lambda^2 \leq \|b
(a^\top)^{-1}\|^2 + \lambda - \lambda^2
$$
It follows that $\mathcal U$ is contained, e.g., in the set $\|\theta\|\leq
\lambda^*$ with 
$$
\lambda^* = \max\{2, \|b (a^\top)^{-1}\bar \theta\|\sqrt{2}\}. 
$$
\end{remark}

\section{Long-time large deviations for the Wishart volatility model}\label{sec:LDP}

In this section, we prove that the Wishart stochastic volatility model satisfies a large deviation principle when time tends to infinity. 

\subsection{Reminder of large deviations theory}

Let us recall some standard definitions and results of large
deviations theory. For a wider overview of large deviations theory, we
refer the reader to \cite{De:Ze}. We consider a family
$(X^\epsilon)_{\epsilon >0}$ of random variables on a
measurable space $(\X, \B)$, where $\X$ is a topological space. 
\begin{definition}[Rate function]
A rate function $\Lambda^*$ is a lower semi-continuous mapping $\Lambda^*: \X \rightarrow [0,\infty]$. A good rate function is a rate function such that, for every $a \in [0,\infty]$, $\{x : \Lambda^*(x) \le a\}$ is compact. 
\end{definition}
\begin{definition}[Large deviation principle]
$(X^\epsilon)_{\epsilon>0}$ satisfies a large deviation principle with rate function $\Lambda^*$ if, for every $A \in \B$, denoting $\overset{\circ}{A}$ and $\bar{A}$ the interior and the closure of $A$, 
\[
-\inf_{x \in \overset{\circ}{A}} \Lambda^*(x) \le \liminf_{\epsilon \rightarrow 0} \epsilon \log \P(X^\epsilon \in A) \le \limsup_{\epsilon \rightarrow 0} \epsilon \log \P(X^\epsilon \in A) \le -\inf_{x \in \bar{A}} \Lambda^*(x) \:.
\]
\end{definition}
\begin{definition}
Let $f \, : \R^n \rightarrow \R \cup \{+\infty\}$ be a 
convex function with domain $\D:=\{x\in \R^n: f(x)<\infty\}$. $f$ is
called \textbf{\textit{essentially smooth}} if $f$ is differentiable
on $\overset{\circ}{\D} \neq \emptyset$ and for every $x \in \bar{\D}
\backslash \overset{\circ}{\D}$, $\lim_{y \rightarrow x} ||\nabla
f(y)|| = +\infty$.
\end{definition}
The following theorem is the celebrated G\"artner-Ellis theorem of the
large deviations theory. \cite{De:Ze} give a version of this
theorem for a family of random variables parameterized by an integer
number (see paragraph 2.3 in their book), but the version for families
parameterized by a real number is easily deduced from the abstract
G\"artner-Ellis theorem given in paragraph 4.5.3. 
\begin{theorem}[G\"artner-Ellis]\label{thm:GartnerEllis}
Let $\left(X^\epsilon\right)_{\epsilon >0}$ be a family of random vectors in $\R^n$. Assume that for each $\lambda \in \R^n$,
\begin{equation}\label{eq:LambdaDefinition}
\Lambda(\lambda) := \lim_{\epsilon \rightarrow 0} \epsilon \log \E\left[ e^{\frac{\q{\lambda,X^\epsilon}}{\epsilon}}\right] 
\end{equation}
exists as an extended real number. Assume also that 0 belongs to the interior of $D_\Lambda:=\{\lambda \in \R^n \,:\, \Lambda(\lambda) < \infty\}$. Denoting 
\[
\Lambda^*(x) = \sup_{\lambda\in \R^n} \q{\lambda,x} - \Lambda(\lambda)\:, 
\]
the Fenchel-Legendre transform of $\Lambda$, the following hold.
\begin{enumerate}[(a)]
\item For any closed set $F$,
\[
\limsup_{\epsilon \rightarrow 0} \epsilon \, \log \P(X^\epsilon \in F) \le -\inf_{x \in F} \Lambda^*(x) \:.
\]
\item For any open set $G$,
\[
\liminf_{\epsilon \rightarrow 0} \epsilon \, \log \P(X^\epsilon \in G) \ge -\inf_{x \in G \cap \F} \Lambda^*(x) \:,
\]
where $\F$ is the set of exposed points of $\Lambda^*$, whose exposing hyperplane belongs to the interior of $D_\Lambda$.
\item If $\Lambda$ is an essentially smooth, lower semi-continuous function, then $(X^\epsilon)_{\epsilon>0}$ satisfies a large deviations principle with good rate function $\Lambda^*$. 
\end{enumerate}
\end{theorem}
\begin{remark}\label{rem:ConvexityOfLambda} 
The function $\Lambda$ of \eqref{eq:LambdaDefinition} is a convex function. Indeed, let $\lambda, \mu \in \R^n$ and $u \in (0,1)$. A direct application of H\"older's inequality yields
\[
\E\left[ e^{\frac{\q{u \lambda + (1-u)\mu , X^\epsilon}}{\epsilon}}\right] = \E\left[ e^{\frac{\q{u \lambda, X^\epsilon}}{\epsilon}} e^{\frac{\q{(1-u)\mu , X^\epsilon}}{\epsilon}}\right] \le \left(\E\left[ e^{\frac{\q{\lambda, X^\epsilon}}{\epsilon}}\right]\right)^u \left(\E\left[e^{\frac{\q{\mu , X^\epsilon}}{\epsilon}}\right]\right)^{1-u} \:.
\]
Applying the logarithm then proves that $\lambda \mapsto \log \E\left[ e^{\frac{\q{\lambda,X^\epsilon}}{\epsilon}}\right]$ and therefore $\Lambda$ are convex.
\end{remark}
\begin{theorem}[Varadhan's Lemma, extension of \cite{Gu:Ro2008}]\label{thm:Varadhan}
Let $(\X,\B)$ be a metric space with its Borel $\sigma$-field. Let $(X^\epsilon)_{\epsilon >0}$ be a family of $\X$-valued random variables that satisfies a large deviations principle with rate function~$\Lambda^*$. If $\varphi: \X \rightarrow \R \cup \{-\infty\}$ is a continuous function which satisfies 
\[
\limsup_{\epsilon \rightarrow 0} \epsilon\,\log \E\left[\exp\left(\frac{\gamma\,\varphi(X^\epsilon)}{\epsilon}\right) \right] < \infty
\]
for some $\gamma > 1$, then, for any $A\in \B$, 
\begin{multline*}
\sup_{x\in A^\circ} \{ \varphi(x)-\Lambda^*(x) \} \leq
\liminf_{\epsilon\rightarrow 0} \epsilon \,\log\int_{A^\circ}
\exp\left( \frac{\varphi(z)}{\epsilon} \right)
d\mu_\epsilon(z) \\\leq
\limsup_{\epsilon\rightarrow 0} \epsilon \,\log\int_{\bar A}
\exp\left( \frac{\varphi(z)}{\epsilon} \right)
d\mu_\epsilon (z)= \sup_{x\in\bar A} \{ \varphi(x)-\Lambda^*(x) \}\:,
\end{multline*}
where $\mu^\epsilon$ denotes the law of $X^\epsilon$
\end{theorem}

\subsection{Long-time behaviour of the Laplace transform of the log-price}

Let $T > 0$ and define the transformation $Y_T^\epsilon := \epsilon Y_{T/\epsilon}$, which corresponds to the long-time behaviour of $Y_T$. We are interested in the function 
\[
\theta \mapsto \lim_{\epsilon \rightarrow 0} \epsilon \, \log\E\left[e^{\epsilon^{-1}\theta^\top Y_T^\epsilon}\right]\:.
\]
We first give the following lemma.
\begin{lemma}\label{lemma:MatrixBound}
Let $A,B \in \M_n$ such that $A+t B$ est invertible for all
$t\geq t_0$. Then, $(A+t B)^{-1} tB$ is bounded for all sufficiently
large $t$. 
\end{lemma}
\begin{proof}
Since $A+t_0 B$ is invertible, for all $t\geq t_0$, 
$$
(A+t B)^{-1} tB = \Big\{I  +(t-t_0) \tilde B\Big\}^{-1} (t-t_0)
\tilde B \frac{t}{t-t_0},
$$
where $\tilde B = (A+t_0B)^{-1} B$. Now, the fact that $A+tB$ est
invertible for $t\geq t_0$ means that the eigenvalues $\lambda_i$ of
$\tilde B$ satisfy $\lambda_i>0$ or $\Im \lambda_i \neq 0$ for all
$i$. { This implies $\det[I  +(t-t_0) \tilde B]\underset{t\rightarrow +\infty}\sim c t^n$ for some $c\not =0$, and since the adjugate matrix of $I  +(t-t_0) \tilde B$ has coefficients of order $\mathcal{O}(t^{n-1})$, we get that $\Big\{I  +(t-t_0) \tilde B\Big\}^{-1}$ is bounded for $t\ge t_0$.} Therefore, 
$\Big\{I  +(t-t_0) \tilde B\Big\}^{-1} (t-t_0)
\tilde B = I-\Big\{I  +(t-t_0) \tilde B\Big\}^{-1}$ is bounded, and $(A+t B)^{-1} tB$ as well, whenever $t$ is sufficiently large. 
\end{proof}
We now characterise the asymptotic behaviour of the Laplace transform of $Y_t^\epsilon$.
\begin{proposition}\label{prop:ExponentialCondition}
Define 
\begin{equation}\label{eq:Lambda}
\Lambda(\theta) := 
\begin{cases}
T\left( r \, \theta^{\!\top} \! \mathbf{1} - \frac{\alpha}{2} \, \tr{b + \phi^{1/2}(\theta)} \right) & \text{ if } \theta \in \U \\
\infty & \text{ if } \theta \not\in \U 
\end{cases} \:.
\end{equation}
For every $\theta \in \U$,
\[
\lim_{\epsilon \rightarrow 0} \epsilon \log \E\left[e^{\epsilon^{-1}\,\theta^\top Y_T^\epsilon}\right] = \Lambda(\theta) \:.
\]
\end{proposition}
\begin{proof}
Let $\theta \in \U$. By Proposition \ref{prop:LaplaceExplicitForm}, 
\begin{equation}\label{eq:LaplaceTransformEpsilon}
\begin{aligned}
\epsilon \log \E\left[e^{\epsilon^{-1}\,\theta^\top Y_T^\epsilon}\right] 
& = \epsilon \log \E\left[e^{\,\theta^\top Y_{T/\epsilon}}\right] \\
& = \epsilon \left(\theta^\top Y_0 - \frac{1}{2} \tr{\left(b + \phi^{1/2}(\theta)\right) x} \right) \\
& \qquad + \frac{1}{2} \epsilon \tr{\exp\left(- T/\epsilon \, \phi^{1/2}(\theta)\right) \left( b + \phi^{1/2}(\theta)\right) V^{-1}(T/\epsilon)\, x} \\
& \qquad + T \, r \theta^\top \mathbf{1} - \frac{T\,\alpha}{2}\tr{b} - \frac{\alpha}{2}\epsilon \log \det\left[ V(T/\epsilon)\right] \:.
\end{aligned}
\end{equation}
Write $\phi(\theta) = P D P^\top$, where $D$ is diagonal, $P$ is
orthonormal and let $\hat{b} = - P^\top b \, P \in \SPP$. Then
\[
V(t) = P \left( \cosh\left(t \, D^{1/2}\right) + \sinh\left(t \, D^{1/2}\right) D^{-1/2} \, \hat{b}\right) P^\top \:, 
\]
Let $\Ec$ and $\tilde{\Ec}$ be $n\times n$ square matrices with $\Ec_{ij} = \mathds{1}_{\{i=j, D_{ii} = 0\}}$ and $\tilde{\Ec}_{ij} =  D_{ii}^{-1/2} \mathds{1}_{\{i=j,D_{ii} \neq 0\}}$. We then have 
\[
\cosh\left(t \, D^{1/2}\right) = \frac{e^{t D^{1/2}}}{2}\left(I_n +
  e^{- 2 t D^{1/2}}\right) = \frac{e^{t D^{1/2}}}{2}\left(I_n + \Ec +
  \smallO{t^{-1}} \right)
\]
and 
\[
\sinh\left(t \, D^{1/2}\right) D^{-1/2} = \frac{e^{t D^{1/2}}}{2} D^{-1/2}\left(I_n - e^{- 2 t D^{1/2}}\right) = \frac{e^{t D^{1/2}}}{2}\left(\tilde{\Ec} + 2t \Ec + \smallO{t^{-1}} \right) \:.
\]
Therefore,
\begin{equation}\label{eq:V(t)}
\begin{aligned}
V(t) 
& = \frac{1}{2} \, P e^{t D^{1/2}} \left( \left(I_n + \Ec \right) + (2t \Ec + \tilde{\Ec}) \, \hat{b} + \smallO{t^{-1}} \right) P^\top \\
& = - \frac{1}{2} \, P \left(I_n + \Ec \right) e^{t D^{1/2}} \left( \hat{b}^{-1} + (t \, \Ec + \tilde{\Ec}) + \smallO{t^{-1}} \right) P^\top b
\end{aligned}
\end{equation}
and
\[
V^{-1}(t) 
= - 2 \, b^{-1} P \left( \hat{b}^{-1} + (t \, \Ec + \tilde{\Ec}) + \smallO{t^{-1}} \right)^{-1} e^{-t D^{1/2}} \left(I_n - \frac{1}{2} \Ec \right) P^\top
\]
where the invertibility of  $\left(\hat{b}^{-1} + (t \, \Ec + \tilde{\Ec}) + \smallO{t^{-1}} \right)$ is guaranteed for every $t \ge 0$ by the existence of the Laplace transform. Since $\hat{b}^{-1}\in \SPP$ and $(t \, \Ec + \tilde{\Ec})\in \SP$, $\hat{b}^{-1} + (t \, \Ec + \tilde{\Ec})\in \SPP$ and is therefore invertible. Hence
\[
\left(\hat{b}^{-1} + (t \, \Ec + \tilde{\Ec}) + \smallO{t^{-1}} \right) = \left(\hat{b}^{-1} + (t \, \Ec + \tilde{\Ec}) \right) \left(I_n + \smallO{t^{-1}} \right)
\]
and 
\[
V^{-1}(t) 
= - 2 \, b^{-1} P \left(I_n + \smallO{t^{-1}} \right) \left( \hat{b}^{-1} + (t \, \Ec + \tilde{\Ec}) \right)^{-1} e^{-t D^{1/2}} \left(I_n - \frac{1}{2} \Ec \right) P^\top \:.
\]
But
\begin{align*}
& \left(\hat{b}^{-1} + \left(t \, \Ec + \tilde{\Ec} \right)\right)^{-1} \!\! e^{-t D^{1/2}} \\ 
& \hspace{1cm} = \left(\hat{b}^{-1} + \left(t \, \Ec + \tilde{\Ec} \right)\right)^{-1} \left(\Ec + (I_n-\Ec) \right) \, e^{-t D^{1/2}} \\
& \hspace{1cm} = t^{-1} \left(\hat{b}^{-1} + \left(t \, \Ec + \tilde{\Ec} \right)\right)^{-1} \!\! t \Ec + \left(\hat{b}^{-1} + \left(t \, \Ec + \tilde{\Ec} \right)\right)^{-1} (I_n-\Ec) \, e^{-t D^{1/2}} \:,
\end{align*}
where $\left(\hat{b}^{-1} + \left(t \, \Ec + \tilde{\Ec} \right)\right)^{-1} \!\! t \Ec$ is bounded by Lemma \ref{lemma:MatrixBound}. Therefore, 
\[
\left(\hat{b}^{-1} + \left(t \, \Ec + \tilde{\Ec} \right)\right)^{-1} \!\! e^{-t D^{1/2}} \rightarrow 0
\]
and $V^{-1}(t) \rightarrow 0$ as $t \rightarrow \infty$. {Using \eqref{eq:V(t)}, we find
\begin{align*}
\epsilon \log \det\left[ V(T/\epsilon)\right] 
& = T \, \tr{D^{1/2}} + \epsilon \log \det\left[\frac{1}{2} \left(I_n \!+ \Ec \right) \left( I_n + (\epsilon^{-1} T \Ec \!+ \tilde{\Ec}) \hat{b} + \smallO{\epsilon} \right)\right] \\
& = T \, \tr{\phi^{1/2}(\theta)} + \epsilon \log \det\left[  \epsilon^{-1} T \Ec \hat{b} +  \frac{1}{2} \left(I_n \!+ \Ec \right) \left( I_n + \tilde{\Ec} \hat{b} \right) + \smallO{\epsilon}\right] \\
& = T \, \tr{\phi^{1/2}(\theta)} - n \epsilon \log(\epsilon)  + \epsilon \log \det\left[ T \Ec \hat{b} +  \frac{\epsilon}{2} \left(I_n + \Ec + \tilde{\Ec} \hat{b} \right) + \smallO{\epsilon^2} \right].
\end{align*}
We have $\det\left[ T \Ec \hat{b} +  \frac{\epsilon}{2} \left(I_n + \Ec + \tilde{\Ec} \hat{b} \right) + \smallO{\epsilon^2} \right] \sim_{\epsilon \rightarrow 0}  \det\left[ T \Ec \hat{b} +  \frac{\epsilon}{2} \left(I_n + \Ec + \tilde{\Ec} \hat{b} \right) \right]$, since the latter determinant is a non-zero polynomial of $\epsilon$ (for $\epsilon=2T$ the determinant is clearly positive).
}
Thus, by passing to the limit, $\lim_{\epsilon \rightarrow 0}\epsilon \log \det\left[ V(T/\epsilon)\right] = T \, \tr{\phi^{1/2}(\theta)}$. Furthermore, since $ \phi \in \SP$, $\exp\left(- \frac{T}{\epsilon} \, \phi^{1/2}(\theta)\right)$ is bounded. Therefore, 
\[
\tr{\exp\left(- \frac{T}{\epsilon} \, \phi^{1/2}(\theta)\right) \left( b + \phi^{1/2}(\theta)\right) V^{-1}(T/\epsilon)\, x} \underset{\epsilon \rightarrow 0}{\longrightarrow} 0 \:.
\]
Finally, passing to the limit in \eqref{eq:LaplaceTransformEpsilon} finishes the proof.
\end{proof}
The next proposition proves the essential smoothness of $\Lambda$.
\begin{proposition}\label{prop:EssentialSmoothness}
The function $\theta \mapsto \Lambda(\theta)$ defined in \eqref{eq:Lambda} is essentially smooth.
\end{proposition}
\begin{proof}
The function $\Lambda$ defined in \eqref{eq:Lambda} is a lower semi-continuous proper convex function with domain $\U$. Furthermore, since for every $\theta \in \overset{\circ}{\U}$, $\phi(\theta) \in \SPP$, $\Lambda$ is of class $C^1$ on $\overset{\circ}{\U}$. Only remains to prove that $||\nabla_\theta\Lambda(\theta)|| \rightarrow \infty$ when $\theta$ goes to the boundary of $\U$. Let $\theta \in \overset{\circ}{\U}$. By Proposition \ref{prop:ExponentialCondition} 
\[
\Lambda(\theta) 
= T\left( r \, \theta^{\!\top} \! \mathbf{1} - \frac{\alpha}{2} \, \tr{b + \phi^{1/2}(\theta)} \right)\:.
\]
Then for every $j \in \{1,...,n\}$,
\[
\de{\theta_j}\Lambda(\theta) 
= T\left( r - \frac{\alpha}{2} \, \tr{\de{\theta_j}\!\left[\phi^{1/2}\right](\theta)} \right) \:,
\]
where $\de{\theta_j}\!\left[\phi^{1/2}\right]\!(\theta)$ {satisfies
  $$\de{\theta_j}\phi(\theta)= \de{\theta_j}\left[\phi^{1/2}(\theta)\phi^{1/2}(\theta)\right]=\phi^{1/2}\!(\theta) \, \de{\theta_j}\!\left[\phi^{1/2}\right]\!(\theta) + \de{\theta_j}\!\left[\phi^{1/2}\right]\!(\theta) \, \phi^{1/2}\!(\theta). $$
  Multiplying this equation by $\phi^{-1/2}\!(\theta)$ and using the cyclic property of the trace, we get
  $$\tr{\de{\theta_j}\!\left[\phi^{1/2}\right](\theta)} = \frac{1}{2} \, \tr{\phi^{-1/2}\!(\theta) \, \de{\theta_j}\phi(\theta) }. $$
}

and therefore
\begin{equation}\label{eq:Gradient}
\de{\theta_j}\Lambda(\theta) = 
T\left( r - \frac{\alpha}{2} \, \tr{\de{\theta_j}\!\left[\phi^{1/2}\right](\theta)} \right) = 
T\left( r - \frac{\alpha}{4} \, \tr{\phi^{-1/2}\!(\theta) \, \de{\theta_j}\phi(\theta) } \right) \:,
\end{equation}
where 
\[
\de{\theta_j}\phi(\theta) 
= a \! \left(e_j e_j^\top - \theta e_j^\top - e_j \theta^\top\right) a^\top \:.
\]
{We write $\phi(\theta)=PDP^\top$ with $D \in \SPP$ diagonal and denote $w= a ^\top P$, which is invertible since $P$ is orthonormal and $a^\top a \in \SPP$. Then
\begin{align*}
\tr{\phi^{-1/2}\!(\theta) \, \de{\theta_j}\phi(\theta) }  
& = \tr{D^{-1/2}P^\top \de{\theta_j}\phi(\theta) \: P } \\
& = \tr{D^{-1/2} \, w^\top \! \left(e_j e_j^\top - \theta e_j^\top - e_j \theta^\top\right) w } \\
& = \tr{D^{-1/2} \, w^\top \! \left(e_j e_j^\top - 2 e_j \theta^\top\right) w } = \sum_{i=1}^n D_{ii}^{-1/2} (w_{ji}^2 - 2 w_{ji}\: (\theta^\top w e_i)) \:.
\end{align*}
Now, we observe that 
\begin{align*}
D_{ii} & = P_i^\top \phi(\theta) \: P_i = ||b \, P_i||^2 + e_i^\top w^\top \left(\diag(\theta) - \theta\theta^\top \right) \, w e_i \\
& = ||b \, P_i||^2 + \sum_{j=1}^n \theta_j w_{ji}^2 - (\theta^\top \, w e_i)^2 \\
& = ||b \, P_i||^2 + (\theta^\top \, w e_i)^2 + \sum_{j=1}^n \theta_j (w_{ji}^2 - 2 w_{ji}\: (\theta^\top w e_i) ). 
\end{align*}
Therefore, we get by the triangular inequality
\begin{align*}
  \sum_{j=1}^n |\theta_j| \left| \tr{\phi^{-1/2}\!(\theta) \, \de{\theta_j}\phi(\theta) } \right|&\ge \left| \sum_{j=1}^n \theta_j\sum_{i=1}^n D_{ii}^{-1/2} (w_{ji}^2 - 2 w_{ji}\: (\theta^\top w e_i)) \right|  \\
  &=\left|\sum_{i=1}^n D_{ii}^{1/2}-D_{ii}^{-1/2}(||b \, P_i||^2 + (\theta^\top \, w e_i)^2)\right|.
\end{align*}

Then, if $\theta \rightarrow  \bar{\theta}$ with $\bar{\theta}\in \U \setminus \overset{\circ}{\U}$, there exists $i$ such that $D_{ii} \rightarrow 0$ and therefore $\sum_{i=1}^n D_{ii}^{1/2}-D_{ii}^{-1/2}(||b \, P_i||^2 + (\theta^\top \, w e_i)^2) \rightarrow -\infty$ since $||b \, P_i||^2 + (\theta^\top \, w e_i)^2\ge \underline{\lambda}(-b)^2 >0$, where $ \underline{\lambda}(-b)$ is the smallest eigenvalue of~$-b\in \SPP$. 
Therefore, $ \left| \tr{\phi^{-1/2}\!(\theta) \, \de{\theta_j}\phi(\theta) } \right| \rightarrow +\infty$ for some $j$, which implies then $|\de{\theta_j}\Lambda(\theta)|\rightarrow +\infty$. Thus,  $||\nabla_\theta \Lambda(\theta) || \rightarrow \infty$ and $\Lambda$ is therefore essentially smooth. }
\end{proof}
\begin{remark}\label{rem:ExponentialCondition2}
Since, by Remark \ref{rem:ConvexityOfLambda}, $\theta \mapsto \lim_{\epsilon \rightarrow 0} \epsilon \,\log \E\left[e^{\epsilon^{-1}\theta^\top Y_t^\epsilon}\right]$ is a convex function, and, by Proposition \ref{prop:EssentialSmoothness}, $\Lambda$ admits infinite derivative on $\U \backslash \overset{\circ}{\U}$, then for every $\theta \in \R^n \backslash \U$, $\lim_{\epsilon \rightarrow 0} \epsilon \,\log \E\left[e^{\epsilon^{-1}\theta^\top Y_t^\epsilon}\right] = \Lambda(\theta) = \infty$. Therefore, Proposition \ref{prop:ExponentialCondition} does not only hold for $\theta \in \U$, but for every $\theta \in \R^n$. 
\end{remark}

\subsection{Long-time large deviation principle for the log-price process}

We now state the large deviation principle for the family $(Y_T^\epsilon)_{\epsilon >0}$, when $\epsilon \rightarrow 0$.
\begin{theorem}\label{thm:LDP}
The family $(Y_T^\epsilon)_{\epsilon >0}$ satisfies a large deviation principle, when $\epsilon \rightarrow 0$ with good rate function 
\[
\Lambda^*(y) = \sup_{\lambda\in \R^n} \q{\lambda,y} - \Lambda(\lambda)\:.
\]
\end{theorem}
\begin{proof}.
First note that $\phi(0) = b^2 \in \SPP$. But since 
\[
\theta \mapsto \phi(\theta) := b^2 + a \, \left(\diag(\theta) - \theta\theta^\top \right) a^\top
\] 
is a continuous function, there exists a neighbourhood $B(0,\delta)$ of 0 such that $\phi(\theta) \in \SPP$ for every $\theta \in B(0,\delta)$, hence $0 \in \overset{\circ}{\U}$. 
Furthermore, Proposition \ref{prop:ExponentialCondition} together with the argument in Remark \ref{rem:ExponentialCondition2} prove that 
\[
\Lambda(\theta) = \lim_{\epsilon \rightarrow 0} \epsilon \, \log\E\left[e^{\epsilon^{-1}\theta^\top Y_T^\epsilon}\right] \:,
\] 
where $\Lambda$ is defined in \eqref{eq:Lambda}. 
Finally, Proposition \ref{prop:EssentialSmoothness} yields the essential smoothness of $\Lambda$. Therefore, by the G\"artner-Ellis Theorem \ref{thm:GartnerEllis}, $(Y_T^\epsilon)_{\epsilon >0}$ satisfies a large deviation principle, when $\epsilon \rightarrow 0$ with good rate function $\Lambda^*$.
\end{proof}

\section{Asymptotic implied volatility of basket options}\label{sec:AsymptoticPricing}
In this section, to simplify the formulas and without loss of
generality, we assume that $Y^j_0 = 0$ for $j=1,\dots,n$ and $r = 0$
so that $(e^{Y^j_t})_{t\geq 0}$  is a martingale with initial value
$1$ (this follows from Proposition \ref{prop:LaplaceExplicitForm}). We are interested in the limiting behavior far from
maturity of basket
option
prices and the corresponding implied volatilities  in the Wishart
model.  The basket call option price with log strike $k$ and time to
maturity $T$ is defined by 
$$
C(T,k) = \E\left[\left(\sum_{i=1}^n \omega_i S^i_T - e^k\right)_+\right],
$$
and the corresponding put option price is defined by 
$$
P(T,k) = \E\left[\left(e^k-\sum_{i=1}^n \omega_i S^i_T \right)_+\right]
$$
where $\omega \in (\R_+)^n$ with $\sum_{i=1}^n \omega_i = 1$. 

The implied volatility of basket options is defined by comparing their
price to the corresponding option price in the Black-Scholes model
$\frac{dS_t}{S_t} = \sigma dW_t$:
$$
C^{BS}(T,k,\sigma) =  N(d_1) - e^k N(d_2), \quad d_{12} = \frac{
  k\pm \frac{1}{2} \sigma^2 T}{\sigma\sqrt{T}},
$$
where $N$ is the standard normal distribution function. The implied
volatility for log strike $k$ and time to maturity $T$ is then defined as
the unique value $\sigma(T,k)$ such that 
$$
C^{BS}\left(T,k,\sigma(T,k)\right) = C(T,k).
$$
It can be equivalently defined using the put option price. 

It is well known that in most models, for fixed log strike $k$, the
implied volatility converges to a constant value independent from $k$
as $T\to \infty$ \cite{tehranchi2009asymptotics}. To obtain a
non-trivial limiting smile, we therefore follow \cite{Ja:Ke:Mi2013}
and use a renormalized log strike $k(T) = yT$. We are interested in
computing the limiting implied volatility 
$$
\sigma_\infty (y) = \lim_{T\to \infty} \sigma(T,yT). 
$$

\subsection{Asymptotic price for the Wishart model}
Introduce the renormalized log-price process in the stochastic
volatility Wishart model: $\tilde Y^j_T = T^{-1} Y^j_T$, $j=1,\dots,n$. Note that to
simplify notation, in this section we avoid using an extra parameter
$\epsilon$ and simply consider the asymptotics when $T\to
\infty$. For this reason, the asymptotic Laplace exponent
$\Lambda(\theta)$ will be given by equation \eqref{eq:Lambda} with $T= 1$
and $r=0$. 

Denote the basket log price by   $\B_T := \log
\sum_{j=1}^n \omega_j e^{Y^j_{T}}$, and the corresponding renormalized
price by $\tilde \B_T := T^{-1}\log
\sum_{j=1}^n \omega_j e^{Y^j_{T}}$. We
first show some LDP-like bounds for this quantity. In the following
lemma and below, we will use the fact that $\Lambda(0)=\Lambda(e_j)=0$, which implies in particular that $\Lambda^*(x)\ge 0$ and $\Lambda^*(x)-x_j\ge 0$ for all $x\in\R^d$. Thus, we let  $x^* = \Lambda'(0)$ and
$\tilde x^*_j = \Lambda'_j(e_j)$ for $j=1,\dots,n$ and introduce three
constants: $\beta^* = \max_j x^*_j$, $\hat \beta^* = \min_j \tilde x^*_j$ and
$\tilde \beta^* = \max_j \tilde x^*_j$. It is easy to see from~\eqref{eq:Gradient} that $x^*_j=-\tilde x^*_j<0$ since $\phi(0)=\phi(e_j)=b^2$ is positive definite and $a$ is invertible. We get $\beta^* < 0 < \hat \beta^* \leq \tilde \beta^*$.
\begin{lemma}\label{bounds.lm}
The following estimates hold for $\tilde \B_T$. 
\begin{enumerate}
\item 
If $\beta<\beta^*$ then 
\begin{align}
\lim_{T\to \infty} \,T^{-1} \log
  \P\left(\tilde \B_T \in (-\infty, \,\beta ] \right) &= -\inf_{x\in
  (-\infty, \beta ]^n} \Lambda^*(x) \notag\\ &= \inf_{\lambda \in \mathbb R^n,
  \lambda_i \leq 0, i=1,\dots,n} \{\Lambda(\lambda) - \beta \langle
  \lambda, \mathbf 1\rangle \}<0;\label{bnd1}
\end{align}
otherwise 
$$
\lim_{T\to \infty} \,T^{-1} \log
  \P\left(\tilde \B_T \in (-\infty, \,\beta ] \right)  = 0.
$$
\item If $\beta \geq \beta^*$ then 
\begin{align}
\lim_{T\to \infty} \,T^{-1} \log
  \P\left(\tilde \B_T \in (\beta,\infty ) \right) =-\inf_{x\notin
  (-\infty, \beta ]^n} \Lambda^*(x) = \max_{i=1,\dots,n} \inf_{\lambda\in
  \mathbb R}\{-\lambda \beta + \Lambda(\lambda e_i)\}, \label{bnd2}
\end{align}
otherwise 
$$
\lim_{T\to \infty} \,T^{-1} \log
  \P\left(\tilde \B_T \in (\beta,\infty ) \right)  = 0. 
$$
In addition if $\beta \geq \beta^*$ and $\beta \neq \tilde x^*_i$ for all $i$, then
$$
\lim_{T\to \infty} \,T^{-1} \log
  \P\left(\tilde \B_T \in (\beta,\infty ) \right) < -\beta. 
$$

\item Let $j\in \{1,\dots,n\}$. Then, 
\begin{align}
 \lim_{T\to \infty} \,T^{-1} \log
   \E\left[e^{Y^j_{T}}\mathds 1_{\tilde \B_T \in (-\infty,\beta
   ]} \right]  &= -\inf_{x\in (-\infty, \beta ]^n}
   \Lambda^*(x)-x_j\notag\\
& = \beta + \inf_{\lambda^j \leq 1, \lambda^i\leq 0, i\neq
  j}\{\Lambda(\lambda) - \beta \langle \lambda, \mathbf 1\rangle\}.\label{bnd3}
\end{align}
In addition, if $\tilde x^*_j > \beta$ then 
$$
\lim_{T\to \infty} \,T^{-1} \log
   \E\left[e^{Y^j_{T}}\mathds 1_{\tilde \B_T \in (-\infty,\beta
   ]} \right] <0. 
$$
\item Let $j\in \{1,\dots,n\}$ and assume $\beta>\tilde x_j^*$. Then, \begin{align}
 \lim_{T\to \infty} \,T^{-1} \log
   \E\left[e^{Y^j_{T}}\mathds 1_{\tilde \B_T \in (\beta,\infty
   )} \right]  &= -\inf_{x\notin (-\infty, \beta ]^n}
   \Lambda^*(x)-x_j\notag\\
&=\max_{i=1,\dots,n} \inf_{\lambda \in \mathbb R} \{-\lambda \beta +
  \Lambda(\lambda e_i + e_j)\} < 0. \label{bnd4}
\end{align}
\end{enumerate}
\end{lemma}
\begin{proof}
\begin{enumerate}
\item Since $\omega_{\min}e^{\max_j Y^j_T}\le \sum_{j=1}^n\omega_je^{Y^j_T}\le n\omega_{\max}e^{\max_j Y^j_T}$ with $(\omega_{\min},\omega_{\max}) := \left(\min_{j=1,...,n}
  \omega_j, \max_{j=1,...,n} \omega_j \right)$, we have for every $T>0$ and $\beta \in \R$,
\begin{align*}
\left(\tilde Y_T \in (-\infty, \beta - T^{-1}\log (n\,\omega_{\max}))^n\right) 
& \subset (\tilde \B_T < \beta) \\
& \subset \left( \tilde Y_T \in (-\infty, \beta - T^{-1}\log \omega_{\min})^n\right) \:.
\end{align*}
Therefore, we get for
every $\delta>0$ and $T$ sufficiently large, 
\[
\P\left(\tilde Y_T \in (-\infty, \beta - \delta)^n\right) \le \P(\tilde \B_T < \beta) \le \P\left( \tilde Y_T \in (-\infty, \beta +\delta)^n\right)\:.
\]
Passing to the $\limsup$ and $\liminf$, we get:
\begin{multline*}
\liminf_{T\to \infty} T^{-1} \log \P\left(\tilde Y_T \in (-\infty, \beta - \delta)^n\right) \le \liminf_{T\to \infty} T^{-1} \log \P(\tilde \B_T < \beta) \\\le \limsup_{T\to \infty} T^{-1} \log \P(\tilde \B_T < \beta) \le \limsup_{T\to \infty} T^{-1} \log \P\left( \tilde Y_T \in (-\infty, \beta +\delta)^n\right)\:.
\end{multline*}
Using the large deviations principle for $\tilde Y_T$ (Theorem~\ref{thm:LDP}) further
yields:
\begin{multline*}
-\inf_{x\in (-\infty, \beta - \delta)^n} \Lambda^*(x)\le \liminf_{T\to \infty} T^{-1} \log \P(\tilde \B_T < \beta) \\\le \limsup_{T\to \infty} T^{-1} \log \P(\tilde \B_T < \beta) \le -\inf_{x\in (-\infty, \beta + \delta]^n} \Lambda^*(x),
\end{multline*}
and making $\delta$ tend to zero, we see that 
\begin{multline*}
-\inf_{x\in (-\infty, \beta )^n} \Lambda^*(x)\le \liminf_{T\to \infty} T^{-1} \log \P(\tilde \B_T < \beta) \\\le \limsup_{T\to \infty} T^{-1} \log \P(\tilde \B_T < \beta) \le -\inf_{x\in (-\infty, \beta ]^n} \Lambda^*(x).
\end{multline*}
The fact that the domain of $\Lambda$ is bounded (Remark
\ref{rem:bounded}) implies that $\Lambda^*$ is locally bounded from above and
therefore continuous. The first equality of~\eqref{bnd1} then follows by
continuity of $\Lambda^*$. The second equality then follows from the
definition of $\Lambda^*$ and the minimax theorem (see, e.g., Corollary 37.3.2 in \cite{rockafellar1970convex})
which can be applied because the domain of $\Lambda$ is bounded (cf.~Remark
\ref{rem:bounded}). Finally, the inequality follows from the fact that
the function $f(\lambda) = \Lambda(\lambda) - \beta\langle \lambda,
\mathbf 1 \rangle$ satisfies $f(0) = 0$ and $f'(0) = x^* - \beta\mathbf
1$. Under the condition $\beta<\beta^*$ at least one component of the
derivative is strictly positive, and hence the minimum of $f$ over the
set $\{\lambda_i \leq 0, i=1,\dots,n\}$ is strictly negative. 
\item The first equality in~\eqref{bnd2} follows similarly to the
  previous item. If $\beta< \beta^*$ then $x^*\notin (-\infty, \beta]^n$ and the
  infimum equals $0$. Otherwise by convexity of $\Lambda^*$ the
  infimum is attained on the boundary of this set. Therefore, we can write:
\begin{align*}
-\inf_{x\notin (-\infty,\beta]^n} \Lambda^*(x) & = \max_{i=1,\dots,n} \sup_{x\in \mathbb R^n: x_i = \beta}\{-
                                                       \Lambda^*(x)\}
  \\
& = \max_{i=1,\dots,n}\sup_{x\in \mathbb R^n: x_i = \beta}
                                                       \inf_{\lambda
  \in \mathbb R^n}
  \{-\langle \lambda, x\rangle + \Lambda(\lambda)\}\\
&=\max_{i=1,\dots,n}
\inf_{\lambda \in \mathbb R} \{-\lambda \beta + \Lambda(\lambda e_i)\},
\end{align*}
since the $\inf$ and $\sup$ may once again be interchanged in virtue of the minimax
theorem and then the supremum on $x\in\R^n$ such that $x_i=\beta$ is clearly $+\infty$ when there is $j\not=i$ such that $\lambda_j\not=0$.
Consider the function $f_i:\mathbb R\to \mathbb R$, $f_i(\lambda) = -\lambda \beta + \Lambda(\lambda
e_i)$. 
Since $f_i(1) = - \beta$ and $f'_i(1) = -\beta + \tilde x^*_i$, it follows
that 
$$
\beta + \max_{i=1,\dots,n}
\inf_{\lambda \in \mathbb R} \{-\lambda \beta + \Lambda(\lambda e_i)\} < 0. 
$$
when  $\beta \not = \tilde x^*_i$ for all $i$. 
\item For the first identity in \eqref{bnd3}, remark that, similarly to the first part,
for $T$ sufficiently large, all $\delta >0$ and
$\beta\in \mathbb R$ we have,
\begin{align*}
\E[e^{Y^j_{T}}\mathds 1_{\{\tilde Y_T \in (-\infty, \beta - \delta]^n\} }]
 \geq \E[e^{Y^j_{T}}\mathds 1_{\{\tilde \B_T \leq \beta\}}] \geq \E[e^{Y^j_{T}}\mathds1_{\{ \tilde Y_T \in
  (-\infty, \beta +\delta]^n\}}] \:,
\end{align*}
We can apply Theorem \ref{thm:Varadhan} with the function $H: x\mapsto x_j$ since $\Lambda(e_j)=0$ and $\Lambda(\gamma e_j)<\infty$ for $\gamma>1$ small enough. When $\delta$ goes to zero, we get
\begin{multline*}
\sup_{x\in (-\infty, \beta)^n} \{x_j - \Lambda^*(x)\}\le
\liminf_{T\to \infty} T^{-1} \log \E[e^{Y^j_{T}}\mathds 1_{\{\tilde \B_T \leq \beta\}}]
\\\le \limsup_{T\to \infty} T^{-1} \log \E[e^{Y^j_{T}}\mathds 1_{\{\tilde \B_T \leq \beta\}}] \le \sup_{x\in (-\infty, \beta]^n} \{x_j - \Lambda^*(x)\}.
\end{multline*}
By continuity of $\Lambda^*$, the lower and the upper bounds are equal. Since $\Lambda^*(x)=\sup_{\lambda\in \R^n} \q{\lambda+e_j,x} - \Lambda(\lambda+e_j)$, we get
$$\sup_{x\in (-\infty, \beta]^n} \{x_j - \Lambda^*(x)\}=\sup_{x\in (-\infty, \beta]^n} \inf_{\lambda\in \R^n}\Lambda(\lambda+e_j)-\q{\lambda,x}. $$
 The second identity in~\eqref{bnd3} then
follows from the minimax theorem as above. Finally, to show the
inequality, remark that 
$$
\inf_{\lambda^j \leq 1, \lambda^i\leq 0, i\neq
  j}\{\Lambda(\lambda) - \beta\langle \lambda, \mathbf 1\rangle\} \leq
\inf_{\lambda \leq 1} f_j(\lambda)
$$
and $f'_j(1) = \tilde x^*_j-\beta>0$.
\item The first identity in~\eqref{bnd4} follows as in item~(3). We have $\Lambda^*(x)-x_j\ge 0$ and $\Lambda^*(\Lambda'(e_j))=\Lambda_j'(e_j)=\tilde{x}^*_j$ since $e_j$ is a critical point of $\lambda \mapsto \q{\lambda ,\Lambda'(e_j)}-\Lambda(\lambda)$.  Since $\beta>\tilde{x}^*_j$ and $\Lambda'(e_j)\not \in (-\infty,\beta]^n$, the supremum is attained as in item~(2) on the boundary:
  $$\sup_{x\in \R^n} x_j-\Lambda^*(x)=\max_{i=1,\dots,n}\sup_{x\in \R^n:x_i=\beta} x_j-\Lambda^*(x)=\max_{i=1,\dots,n}\sup_{x\in \R^n:x_i=\beta} \inf_{\lambda\in \R^n}\Lambda(\lambda+e_j)-\q{\lambda,x}. $$
 The
  second identity in~\eqref{bnd4} holds true in virtue of the minimax theorem as
  above, like in item~(2). To prove the negativity, we consider the functions $g_i(\lambda) =  -\lambda \beta + \Lambda(\lambda e_i + e_j)$. We have that $g_i(0) = 0$
  and $g_i'(0) = -\beta + \Lambda'_i(e_j)$. We have $g_j'(0)=-\beta+\tilde x_j^*<0$. If $g'_i(0)\not = 0$ for all $i$, the result is clear. Otherwise, we can find $\tilde \beta \in (\tilde x_j^*,\beta)$ such that $\tilde \beta \not = \Lambda'_i(e_j)$ for all $i$, and since $e^{Y^j_{T}}\mathds 1_{\tilde \B_T \in (\beta,\infty
   )} \le e^{Y^j_{T}}\mathds 1_{\tilde \B_T \in (\tilde\beta,\infty
   )} $, we get the claim. 
\end{enumerate}
\end{proof}

The following theorem characterizes the asymptotic behavior of basket call
prices in the Wishart model. There are different asymptotic
regimes to consider, depending on the position of $y$ with respect to
the constants $\beta^*$, $\tilde \beta^*$ and $\hat \beta^*$. 
\begin{theorem}Assume that $y \neq \tilde x^*_i$ for all $i$. Then,
  as $T\to \infty$, the call option price in the Wishart model satisfies
\begin{align}
\lim_{T\to \infty} \E\left[(
  e^{\mathcal B_{T}}-e^{yT})_+\right]  =
\sum_{i=1}^{n} \omega_i \mathds 1_{\tilde x_i^* > y}.   \label{limitcall}
\end{align}
In addition, if $y< \beta^*$ then 
\begin{align}
\lim_{T\to \infty} T^{-1} \log \E\left[(e^{yT} -
  e^{\mathcal B_{T}})_+\right] &= \lim_{T\to \infty}
T^{-1} \log \left\{e^{yT} - 1+ \E\left[(
  e^{\mathcal B_{T}}-e^{yT})_+\right]\right\}\notag\\ &=  y - \inf_{z \in (-\infty,y]^n } \Lambda^*(z)<y;\label{ldput}
\end{align}
if $y>\tilde \beta^*$, then
\begin{align}
\lim_{T\to \infty} T^{-1} \log \E\left[(
  e^{\mathcal B_{T}}-e^{yT})_+\right] =
\max_{i,j=1,\dots,n} \inf_{\lambda
  \in \mathbb R}
  \{- \lambda y  + \Lambda(\lambda e_i + e_j)\}<0. \label{ldcall}
\end{align}
and if $y\in (\beta^*,\hat \beta^*)$, then
\begin{align}
\lim_{T\to \infty} T^{-1} \log \left(1
- \E[(e^{\B_{T}}-e^{yT})_+] \right) 
 = y + \max_{i=1,\dots,n}\inf_{\lambda \in \mathbb R} \{-\lambda y + \Lambda(\lambda e_i)\}\}<\min(0,y).\label{ldcovcall}
\end{align}
\end{theorem}
\begin{proof}${}$\\
\noindent\textit{Proof of \eqref{limitcall}.}  We remark that 
\begin{align}
\E\left[(
  e^{\mathcal B_{T}}-e^{yT})_+\right] = \E\left[
  e^{\mathcal B_{T}}\mathds 1_{\tilde\B_T
    > y}\right] - e^{yT}\mathbb P\left[\tilde\B_T >
  y \right]\label{call2terms}
\end{align}
and consider the two terms separately. If $y<0$, the second term
clearly converges to zero. Assume then that $y\geq 0$. Since $\beta^*\leq
0$, by Lemma 
\ref{bounds.lm} part 2, 
$$
\lim_{T\to \infty} \,T^{-1} \log
  e^{yT}\P\left(\tilde \B_T >y \right) <0
$$
This proves that the second
term in \eqref{call2terms} converges to zero. We now focus on the
first term, which satisfies
$$
 \E \left[
  e^{\mathcal B_{T}}\mathds 1_{\tilde\B_T
    > y}\right]  = \sum_{i=1}^n \omega_i \E \left[
  e^{Y^i_{T}}\mathds 1_{\tilde\B_T
    > y}\right].
$$
Fix some $i\in \{1,\dots,n\}$. Then, by Lemma
\ref{bounds.lm} parts 3 and 4, if $y> \tilde x^*_i$ then
$$
\lim_{T\to \infty}\E\left[
  e^{Y^i_{T}}\mathds 1_{\tilde\B_T
    > y}\right]  = 0,
$$
and if $y < \tilde x^*_i$ then 
$$
\lim_{T\to \infty}\E\left[
  e^{Y^i_{T}}\mathds 1_{\tilde\B_T
   \le y}\right]  = 0.
$$
Combining these estimates for different $i$, the proof of
\eqref{limitcall} is complete. 

\smallskip

\noindent \textit{Proof of \eqref{ldput}} The equality 
\[ 
e^{y\,T} (1-e^{-\delta T}) \mathds{1}_{\left\{\tilde \B_T < y - \delta \right\}} \le \left(e^{y\,T} - e^{\B_{T}} \right)_+ \le e^{y\,T} \mathds{1}_{\left\{\tilde \B_T < y \right\}}
\]
holds for every $\delta > 0$ and $T >0$. Then by successively taking the expectation, the logarithm and multiplying by $T^{-1}$, we find
\begin{align*}
& y + T^{-1} \log(1-e^{-\delta T}) + T^{-1} \log\P\left(\tilde \B_T < y-\delta\right) \\
\le \: & T^{-1} \log\E\left[(e^{y\,T}-e^{\B_{T}})_+\right] \le y + T^{-1} \log\P\left( \tilde \B_T < y\right) \:.
\end{align*}
Passing to the limit $T\to \infty$ and using Lemma \ref{bounds.lm}
part 1, the proof is complete. 

\smallskip

\noindent\textit{Proof of \eqref{ldcall}.} We use the inequality
\[ 
e^{\B_{T}} (1-e^{-\delta T}) \mathds{1}_{\left\{y<\tilde \B_T  - \delta \right\}} \le \left(e^{\B_{T} } - e^{y\,T} \right)_+ \le e^{\B_{T}} \mathds{1}_{\left\{y<\tilde \B_T \right\}}.
\]
Consider for instance the upper bound. Taking the expectation and the
logarithm, we obtain $\log \mathbb \E[ e^{\B_{T}}
\mathds{1}_{\left\{\tilde \B_T>y \right\}}] = \log  \sum_{j=1}^n \omega_j\E\left[ e^{Y^j_{T}} \mathds{1}_{\left\{\tilde \B_T> y \right\}}\right]
$ and thus
\begin{align*}
T^{-1}\log \mathbb \E[ e^{\B_{T}}
\mathds{1}_{\left\{\tilde \B_T>y \right\}}] &\le  \max_{j=1,\dots,n} T^{-1} \log
 \E\left[ e^{Y^j_{T}}
   \mathds{1}_{\left\{\tilde \B_T> y \right\}}\right],\\
 T^{-1}\log \mathbb \E[ e^{\B_{T}}
\mathds{1}_{\left\{\tilde \B_T>y+\delta \right\}}] & \ge \max_{j=1,\dots,n} T^{-1}\log
 \E\left[ e^{Y^j_{T}}
   \mathds{1}_{\left\{\tilde \B_T> y+\delta \right\}}\right] +\log(\omega_j)/T.
\end{align*}
The result then follows from Lemma \ref{bounds.lm}, part 4. 

\smallskip

\noindent\textit{Proof of \eqref{ldcovcall}.} We use the following
identity. 
\begin{align*}
1
- \E[(e^{\B_{T}}-e^{yT})_+] &= \E[e^{\B_{T}}- (e^{\B_{T}}-e^{yT})_+]\\&= e^{yT} \P[\tilde\B_T >
y] + \E[e^{\B_{T}}
\mathds 1_{\tilde\B_T \leq y}].
\end{align*}
By Lemma \ref{bounds.lm}, part 2, 
\begin{align*}
\lim_{T\to \infty} T^{-1} \log e^{yT}\P[\tilde\B_T > y]  =y+\max_{i=1,\dots,n}
\inf_{\lambda \in \mathbb R} \{-\lambda y + \Lambda(\lambda e_i)\}
  <0. 
\end{align*}
Consider the function $f_i:\mathbb R\to \mathbb R$, $f_i(\lambda) = -\lambda y + \Lambda(\lambda
e_i)$. Since $f_i(0) = 0$ and $f'_i(0)  = -y + x^*_i<0$, it follows
that also 
$$
y + \max_{i=1,\dots,n}
\inf_{\lambda \in \mathbb R} \{-\lambda y + \Lambda(\lambda e_i)\} <
y. 
$$

On the other hand, by Lemma \ref{bounds.lm}, part 3, 
\begin{align*}
\lim_{T\to \infty} T^{-1} \log\E[e^{\B_{T}}
\mathds 1_{\tilde\B_T \leq y}]&= y+\max_{j=1,\dots,n}\inf_{\lambda^j \leq 1,
      \lambda^i\leq 0, i\neq j} \{ 
    \Lambda(\lambda) - y \langle \lambda, \mathbf 1\rangle\}\\
& \leq y+\max_{j=1,\dots,n}\inf_{\lambda \leq 1} f_j(\lambda). 
\end{align*}
Since, for $y\in (\beta^*,\hat \beta^*)$, $f_j'(0)<0$ and $f_j'(1)>0$, the
infimum is attained on the interval $(0,1)$, and the contribution of
this term is less  than the one of the first term. 
The properties of the logarithm allow to conclude the proof. 
\end{proof}

\subsection{Implied volatility asymptotics}\label{ivasymp.sec}
In the Black-Scholes model with volatility $\sigma$,  we have (see, e.g. \cite{FordeJacquier2011}, Corollary 2.12)
\begin{align*}
\lim_{T\to \infty} T^{-1} \log (C^{BS}(T,yT,\sigma)+
  e^{yT} - 1)&= - \frac{1}{2}\left(\frac{\sigma}{2} -
  \frac{y}{\sigma}\right)^2,\quad  y \leq -\frac{\sigma^2}{2}\\
\lim_{T\to \infty} T^{-1} \log C^{BS}(T,yT,\sigma) &= - \frac{1}{2}\left(\frac{\sigma}{2} -
  \frac{y}{\sigma}\right)^2,\quad  y \geq \frac{\sigma^2}{2}\\
\lim_{T\to \infty} T^{-1} \log \left(1 - C^{BS}(T,yT,\sigma)\right) &= - \frac{1}{2}\left(\frac{\sigma}{2} -
  \frac{y}{\sigma}\right)^2, \quad -\frac{\sigma^2}{2} < y <
  \frac{\sigma^2}{2}. 
\end{align*}
Under the Wishart model, for the basket option, we can write:
\begin{align}
\lim_{T\to \infty} T^{-1} \log \E\left[(e^{yT} -
  e^{\B_{T}})_+\right] &= - L(y),\quad  y \leq \beta^*\label{eqconvopt.eq}\\
\lim_{T\to \infty} T^{-1} \log \E\left[(e^{\B_{T}}
  -e^{yT} )_+\right] &= - L(y),\quad  y \geq \tilde \beta^*\notag\\
\lim_{T\to \infty} T^{-1} \log \left(1 - \E\left[(e^{\B_{T}}
  -e^{yT} )_+\right] \right) &= - L(y), \quad \beta^* < y <
  \hat \beta^*,\notag
\end{align}
where 
\begin{align*}
L(y) &= -y - \inf_{\lambda\in \mathbb R^n: \lambda_i \leq 0, i=1,\dots,n}
\{ \Lambda(\lambda)-y \langle \lambda, \mathbf 1\rangle \},\quad y
\leq \beta^* \\
L(y) &= -\max_{i,j=1,\dots,n}\inf_{\lambda
  \in \mathbb R}
  \{- \lambda y  + \Lambda(\lambda e_i + e_j)\},\quad y\geq \tilde
        \beta^*\\
L(y) &= -y - \max_{i=1,\dots,n}\inf_{\lambda \in \mathbb R} \{-\lambda y + \Lambda(\lambda e_i)\},
        \quad \beta^* < y< \hat
        \beta^*.
\end{align*}

We deduce (see \cite{Ja:Ke:Mi2013} for details) that the limiting implied
volatility of a basket option in the Wishart model is given by
\begin{align}
\sigma_\infty(y) = \sqrt{2}\left(\xi\sqrt{ L(y)+ y} + \eta\sqrt{L(y)} \right), \label{iv}
\end{align}
where $\xi$ and $\eta$ are constants with $\xi^2 = \eta^2 = 1$, which
must be chosen to satisfy the conditions
\begin{align*}
&y \leq -\frac{\sigma_\infty^2(y)}{2} &&\text{if}\quad y \leq
  \beta^*\\
&y \geq \frac{\sigma_\infty^2(y)}{2}  &&\text{if}\quad y \geq
 \tilde \beta^*\\
&-\frac{\sigma_\infty^2(y)}{2}  < y < \frac{\sigma_\infty^2(y)}{2}
   &&\text{if}\quad \beta^* < y < 
  \tilde \beta^*.
\end{align*}
 First of all
remark that by taking $\lambda=0$ and $\lambda = e_i$ it follows that
$L(y)\geq y$ and $L(y) \geq 0$, so that the expressions under the
square root sign are positive. It is easy to see that for $y\leq \beta^*$, these conditions imply
$\xi=-1$ and $\eta=1$ since $b^*<0$ and $-y\le L(y)$, and for $y\geq \tilde \beta^*$ one has $\xi = 1$ and
$\eta = -1$.  For $\beta^*< y< \hat \beta^*$, we still have $|y|\le \max(L(y),L(y)+y)$ and to satisfy the
conditions in this case and $\sigma_\infty(y)>0$, one must take $\xi=\eta = 1$. 

The case when $\hat \beta^* < y< \tilde \beta^*$ requires a specific
treatment. It is characterized by the following proposition.
\begin{proposition}
Let $\hat \beta^* < y < \tilde \beta^*$. Then, $\sigma_\infty(y) =
\sqrt{2y}$ and 
$$
\sigma(T,yT)  = \sqrt{2y} + N^{-1}(C_\infty(y))
T^{-1/2} + \smallO{T^{-1/2}}
$$
as $T\to \infty$, where $C_\infty(y) = \sum_{i=1}^n \omega_i
\mathds 1_{\tilde x^*_i > y}$. 
\end{proposition}
\begin{proof}We follow the arguments of the proof of Theorem 3.3 in
  \cite{jacquier2018implied} with some minor changes. 
The Black-Scholes call option price satisfies
$$
C^{BS}(T,yT,\sigma) = N\left(\frac{-y+\frac{\sigma^2}{2}}{\sigma}
  \sqrt{T}\right) - e^{yT} N\left(\frac{-y-\frac{\sigma^2}{2}}{\sigma}
  \sqrt{T}\right).  
$$
We have by definition of the implied volatility and equation~\eqref{limitcall},
$$
C^{BS}(T,yT,\sigma(t,yT)) = C(T,yT) \underset{T\to +\infty}\to C_\infty(y). 
$$
Since $y>\hat \beta^*>0$, as $T\to \infty$, we get
necessarily $\frac{y+\frac{\sigma(T,yT)^2}{2}}{\sigma(T,yT)}
  \sqrt{T}\to +\infty$. Using the classical bound on the Mills ratio
  $N(-x)\leq x^{-1} \phi(x)$ for $x>0$, where $\phi$ is the standard Gaussian
  density, we have 
$$
e^{yT} N\left(\frac{-y-\frac{\sigma(T,yT)^2}{2}}{\sigma(T,yT)}
  \sqrt{T}\right) \le  \phi\left(\frac{y-\frac{\sigma(T,yT)^2}{2}}{\sigma(T,yT)}
  \sqrt{T}\right) \frac{\sigma(T,yT)}{\left(y+\frac{\sigma(T,yT)^2}{2}\right)\sqrt{T}} \to 0
$$
as $T\to \infty$. Therefore,
\begin{align}
\frac{-y+\frac{\sigma(T,yT)^2}{2}}{\sigma(T,yT)}
  = N^{-1}(C_\infty(y)) T^{-1/2} + \smallO{T^{-1/2}}. \label{limd1}
\end{align}
Consider now the function $f(z) = -\frac{y}{z} + \frac{z}{2}$. Its
inverse which is positive in the neighborhood of zero is given by 
$$
f^{-1}(x) = x+ \sqrt{x^2 + 2y}
$$
Applying $f^{-1}$ to both sides of \eqref{limd1} and neglecting terms
of order $\smallO{T^{-1/2}}$, the proof is complete. 
\end{proof}

\section{Variance reduction} \label{sec:VarianceReduction}

Denote $P(S_T)$ the payoff of a European option on $(S_T^1,...,S_T^n)$. The
price of an option is generally calculated as the expectation
$\E(P(S_T))$ under a certain risk-neutral measure $\P$. When the
number of assets $n$ is low, this expectation may be evaluated by
Fourier inversion, however, when the dimension is large, as in the
case of index options, Monte Carlo is the method of choice. The
standard Monte Carlo estimator of $\E(P(S_T))$ with $N$ samples is given by
$$
\widehat P_N = \frac{1}{N} \sum_{j=1}^N P(S^{(j)}_T),
$$
where $S^{(j)}_T$ are i.i.d.~samples of $S_T$ under the measure
$\P$. The variance of the standard Monte Carlo estimator is given by
$$
\text{Var}[\widehat P_N] = \frac{1}{N} \text{Var}[P(S_T)],
$$
and is often too high for real-time applications. To decrease the
computational time, various variance reduction methods have been
proposed, the most popular being importance sampling.

The importance sampling method is based on the following identity, valid for any probability measure $\Q$, with respect to which $\P$ is absolutely continuous. 
$$
\E[P(S_T)] = \E^{\mathbb Q}\left[\frac{d\mathbb
    P}{d\mathbb Q}P(S_T)\right].
$$
This allows one to define the {importance sampling estimator}
$$
\widehat P^{\mathbb Q}_N := \frac{1}{N} \sum_{j=1}^N \left[\frac{d\mathbb
  P}{d\mathbb Q}\right]^{(j)} P(S^{(j),\Q}_{T}),
$$
where $S^{(j),\Q}_T$ are i.i.d.~samples of $S_T$ under the
measure $\mathbb Q$. For efficient variance reduction, one needs then to find a probability measure $\mathbb Q$
such that $S_T$ is {easy to simulate} under $\mathbb Q$ and the variance 
$$
\text{Var}_{\mathbb Q} \left[ P(S_T)\frac{d\mathbb P}{d\mathbb Q}
\right] = \E^{\P}\left[P(S_T)^2\frac{d\P}{d\Q}\right] -
\E^\P[P(S_T)]^2
$$ 
is considerably smaller than the original variance $\text{Var}_{\mathbb P} \left[P(S) \right]$. 

In this paper we consider the class of measure changes $\left\{\P_\theta\,:\, \theta \in \R^n \right\}$, where
\[
\frac{d\P_\theta}{d\P} = \frac{e^{\theta^\top Y_T}}{\E\left[e^{\theta^\top Y_T}\right]}\:.
\]
To find the optimal variance reduction parameter $\theta^*$, we
therefore need to minimize the variance of the estimator under
$\mathbb Q$, or, equivalently, the expectation
$$
\E^{\P}\left[P(S_T)^2\frac{d\P}{d\P_\theta}\right]. 
$$

\subsection{Asymptotic variance reduction}
Denoting $H(Y_T):=\log P\left(e^{Y_T}\right)$, the optimization problem writes
\begin{equation}\label{eq:OptimizationProblemExact}
\inf_{\theta \in \R^n} \E\left[ \exp\left( 2H(Y_T) - \theta^\top Y_T + \G_1(\theta) \right)\right]\:,
\end{equation}
where
\[
\G_\epsilon(\theta) := \epsilon\log \E\left[e^{\frac{\theta^\top Y_T^\epsilon}{\epsilon} }\right] \:.
\]
Since we cannot compute the minimizer for this expression explicitly,
we instead choose to minimize an asymptotic proxy for the variance,
based on Varadhan's lemma (Theorem \ref{thm:Varadhan}). This proxy is
introduced in the following proposition. 
\begin{proposition}\label{prop:AsymptoticProblem}
Let $H: \mathbb R^n \to \mathbb R\cup \{-\infty\}$ be a continuous
function and $\theta \in \R^n$ be such that there exists $\gamma > 1$ with
\begin{equation}\label{eq:ExponentialBound}
\limsup_{\epsilon \rightarrow 0} \epsilon \log \E\left[ \exp\left\{ \gamma\,\frac{2H(Y_T^\epsilon) - \theta^\top Y_T^\epsilon}{\epsilon}\right\}\right] < \infty \:.
\end{equation}
Then
\begin{align*}
& \lim_{\epsilon \rightarrow 0} \epsilon \log \E\left[ \exp\left\{ \frac{2H(Y_T^\epsilon) - \theta^\top Y_T^\epsilon
+ \G_\epsilon(\theta)}{\epsilon}\right\}\right] \\
& \qquad \qquad \qquad \qquad \qquad = \sup_{y \in \R^n} \left\{ 2H(y) - \theta^\top y - \Lambda^*(y) \right\} + \Lambda(\theta) \:.
\end{align*}
\end{proposition}
\begin{proof}
By Theorem \ref{thm:Varadhan}, 
\begin{equation}\label{eq:VaradhanProof1}
\lim_{\epsilon \rightarrow 0} \epsilon \log \E\left[ \exp\left\{ \frac{2H(Y_T^\epsilon) - \theta^\top Y_T^\epsilon}{\epsilon}\right\}\right] = \sup_{y \in \R^n} \left\{2H(y) - \theta^\top y - \Lambda^*(y)\right\} \:.
\end{equation}
Furthermore, by Proposition \ref{prop:ExponentialCondition}, 
\begin{equation}\label{eq:VaradhanProof2}
\epsilon \log \E\left[ \exp\left\{ \frac{\G_\epsilon(\theta)}{\epsilon}\right\}\right] = \G_\epsilon(\theta) \underset{\epsilon \rightarrow 0}{\longrightarrow} \Lambda(\theta) \:.
\end{equation}
Multiplying \eqref{eq:VaradhanProof1} and \eqref{eq:VaradhanProof2} finishes the proof.
\end{proof}
\begin{remark}
In particular, if $H$ is continuous and bounded from above and $\theta$ is such that $\phi(-\theta) \in \SPP$, condition \eqref{eq:ExponentialBound} is met.
\end{remark}
\begin{definition}
A parameter $\theta^* \in \R^n$ is \textit{asymptotically optimal} if it
achieves the infimum in the minimisation problem
\begin{equation} \label{eq:AsymptoticProblem}
\inf_{\theta \in \R^n} \sup_{y \in \R^n} \left\{ 2H(y) - \theta^\top y - \Lambda^*(y) \right\} + \Lambda(\theta) \:.
\end{equation}
\end{definition}
\begin{theorem}\label{thm:OptimalTheta}
Let $H$ be a concave upper semi-continuous function. Then
\[
\inf_{\theta \in \R^n} \sup_{y \in \R^n} \left\{ 2H(y) - \theta^\top y - \Lambda^*(y) \right\} + \Lambda(\theta) = 2 \inf_{\theta \in \R^n} \left\{ \hat{H}(\theta) + \Lambda(\theta) \right\} \:,
\]
where 
\[
\hat{H}(\theta) = \sup_{y \in \R^n} \left\{ H(y) - \theta^\top y \right\} \:.
\]
Furthermore, if $\theta^*$ minimizes the right-hand side, it also
minimizes the left-hand side.
\end{theorem}
\begin{proof}
We follow the idea of the proof of \cite[Theorem 8]{Ge:Ta}, with some
major simplifications due to the present finite-dimensional setting. By definition of $\Lambda^*$,
\begin{align*}
\inf_{\theta \in \R^n} & \sup_{y \in \R^n} \left\{ 2H(y) - \theta^\top y - \Lambda^*(y) + \Lambda(\theta) \right\} \\
= \: & \inf_{\theta \in \R^n} \sup_{y \in \R^n} \left\{ 2H(y) - \theta^\top y - \sup_{\lambda \in \R^n} \left\{\lambda^\top y - \Lambda(\lambda)\right\} + \Lambda(\theta) \right\} \\
& \quad = \: \inf_{\theta \in \R^n} \sup_{y \in \R^n} \inf_{\lambda \in \R^n} \left\{ 2H(y) - \theta^\top y -\lambda^\top y + \Lambda(\lambda) + \Lambda(\theta) \right\} \:.
\end{align*}
The function
\[
(y, \lambda) \mapsto 2H(y) - \theta^\top y -\lambda^\top y + \Lambda(\lambda) + \Lambda(\theta)
\]
is concave-convex on $\R^n \times \U$ where $\U$ is bounded by Remark~\ref{rem:bounded} and both $\R^n$ and $\U$ are convex. Therefore, by the minimax Theorem for concave-convex functions (see, e.g., Corollary 37.3.2 in \cite{rockafellar1970convex}),
\begin{align*}
\sup_{y \in \R^n} \inf_{\lambda \in \R^n} & \left\{ 2H(y) - \theta^\top y -\lambda^\top y + \Lambda(\lambda) + \Lambda(\theta) \right\}\\
= \: & 
\inf_{\lambda \in \R^n} \sup_{y \in \R^n} \left\{ 2H(y) - \theta^\top y -\lambda^\top y + \Lambda(\lambda) + \Lambda(\theta) \right\} \:.
\end{align*}
This allows us to rewrite
\begin{align}
\inf_{\theta \in \R^n} & \sup_{y \in \R^n} \left\{ 2H(y) - \theta^\top y - \Lambda^*(y) + \Lambda(\theta) \right\} \notag\\
= \: & 
\inf_{\theta \in \R^n} \inf_{\lambda \in \R^n} \sup_{y \in \R^n} \left\{ 2H(y) - \theta^\top y -\lambda^\top y + \Lambda(\lambda) + \Lambda(\theta) \right\} \notag\\
& = \: 
2 \inf_{\theta \in \R^n} \inf_{\lambda \in \R^n} \left\{\hat{H}\left(\frac{\theta+\lambda}{2}\right) + \frac{\Lambda(\lambda) + \Lambda(\theta)}{2} \right\} = 2 \inf_{\theta \in \R^n} \left\{ \hat{H}(\theta) + \Lambda(\theta) \right\} \:,\label{lastline}
\end{align}
where the last equality is justified by the fact that, by convexity, 
\[
\frac{\Lambda(\lambda) + \Lambda(\theta)}{2} \ge \Lambda\left( \frac{\lambda + \theta}{2} \right) 
\]
with equality if $\lambda = \theta$.

To prove the last statement of the theorem, assume that the infimum in
the right-hand side of \eqref{lastline} is attained by
$\theta^*$. Then, using the equality of the right-hand side and the
left-hand side, and taking $\lambda=\theta^*$ in the left-hand side,
we see that the same value $\theta^*$ also attains the infimum in
left-hand side. 
\end{proof}
\begin{remark} Similarly to \cite[Definition 6]{Ge:Ta} and to the
  discussion in Section 4 of \cite{Rob2010}, it can be shown that the
  asymptotically optimal $\theta$ in Theorem \ref{thm:OptimalTheta}
  reaches the asymptotic lower bound of the variance on the log-scale
  over all equivalent measure changes.

Let $\Q \sim \P$ be an equivalent measure change. 
Then by Jensen's inequality
\begin{align*}
\lim_{\epsilon \rightarrow 0 } \epsilon \log\E^\Q\left(e^{\frac{2\,H(Y_T^\epsilon)}{\epsilon}} \left(\frac{d\P}{d\Q}\right)^2\right) 
& \ge 2\,\lim_{\epsilon \rightarrow 0 } \epsilon \log\E^\Q\left(e^{\frac{H(Y_T^\epsilon)}{\epsilon}} \frac{d\P}{d\Q}\right) \\
& = 2\,\lim_{\epsilon \rightarrow 0 } \epsilon \log\E\left(e^{\frac{H(Y_T^\epsilon)}{\epsilon}} \right) \:.
\end{align*}
By Theorem \ref{thm:Varadhan}, the right-hand side is equal to
\begin{align*}
2 \sup_{y \in \R^n} \left\{ H(y) - \Lambda^*(y) \right\} 
& = 2 \sup_{y \in \R^n} \inf_{\theta \in \R^n} \left\{ H(y) - \theta^\top y + \Lambda(\theta) \right\} \\
& = 2 \inf_{\theta \in \R^n} \left\{ \sup_{y \in \R^n} \left\{ H(y) - \theta^\top y \right\} + \Lambda(\theta) \right\} \:,
\end{align*}
where the second equality is obtained by the minimax theorem for concave-convex functions  \cite{rockafellar1970convex}, 
already used in the proof of Theorem \ref{thm:OptimalTheta}. But by the same Theorem \ref{thm:OptimalTheta}, this bound is reached when $\theta$ is asymptotically optimal. 
\end{remark}

\section{Numerical results}\label{sec:Numerics}

\subsection{Long-time implied volatility}
Let us now fix the parameters of the model to the values 
\[
b = -\left(\begin{array}{cc}
1.0 & 0.7 \\
0.7 & 0.7
\end{array}\right)\:, \qquad a = \left(\begin{array}{cc}
0.2 & 0 \\
0 & 0.3
\end{array}\right)\]
and $\alpha = 1.5$, with initial values $S_0= \mathds{1}$ and $x =  I_2$ and consider the problem of pricing a basket put option with log-payoff 
\[
H(Y_T) = \log\left(K -  \frac{1}{2} e^{Y_T^1} + \frac{1}{2} e^{Y_T^2}\right)_+.
\]

Figure \ref{smile.fig} shows the implied volatility smile for such an
option, for $T = \frac{1}{3}$, computed by Monte Carlo over 100,000
trajectories, together with the $95\%$ confidence interval. To sample the paths of the process,
we use the exact simulation of the Wishart process described in~\cite{AA2013}, Algorithm~3. Thus, we obtain the values of
$X_{t_i}$ on the regular time grid~$t_i=i \Delta t$, with $i\in \N$ and $\Delta t>0$. Then, for the stock, we use a trapezoidal rule since it gives a second-order weak convergence (see Section 4.3 in~\cite{AA2013} for details):
 $$Y_{t_{i+1}} = Y_{t_i} - \frac{1}{2}\text{diag}\left[ a^\top \frac{X_{t_i}+X_{t_{i+1}}}{2} a\right] \Delta t + \textup{Chol}\left(a^\top \frac{X_{t_i}+X_{t_{i+1}}}{2} a \right)(Z_{t_{i+1}} - Z_{t_i}),$$
where $Z$ is a Brownian motion sampled independently from $X$ and
$\textup{Chol}(M)$ is the Cholesky decomposition of a positive definite matrix $M$.


\begin{figure}
\centerline{\includegraphics[width=0.6\textwidth]{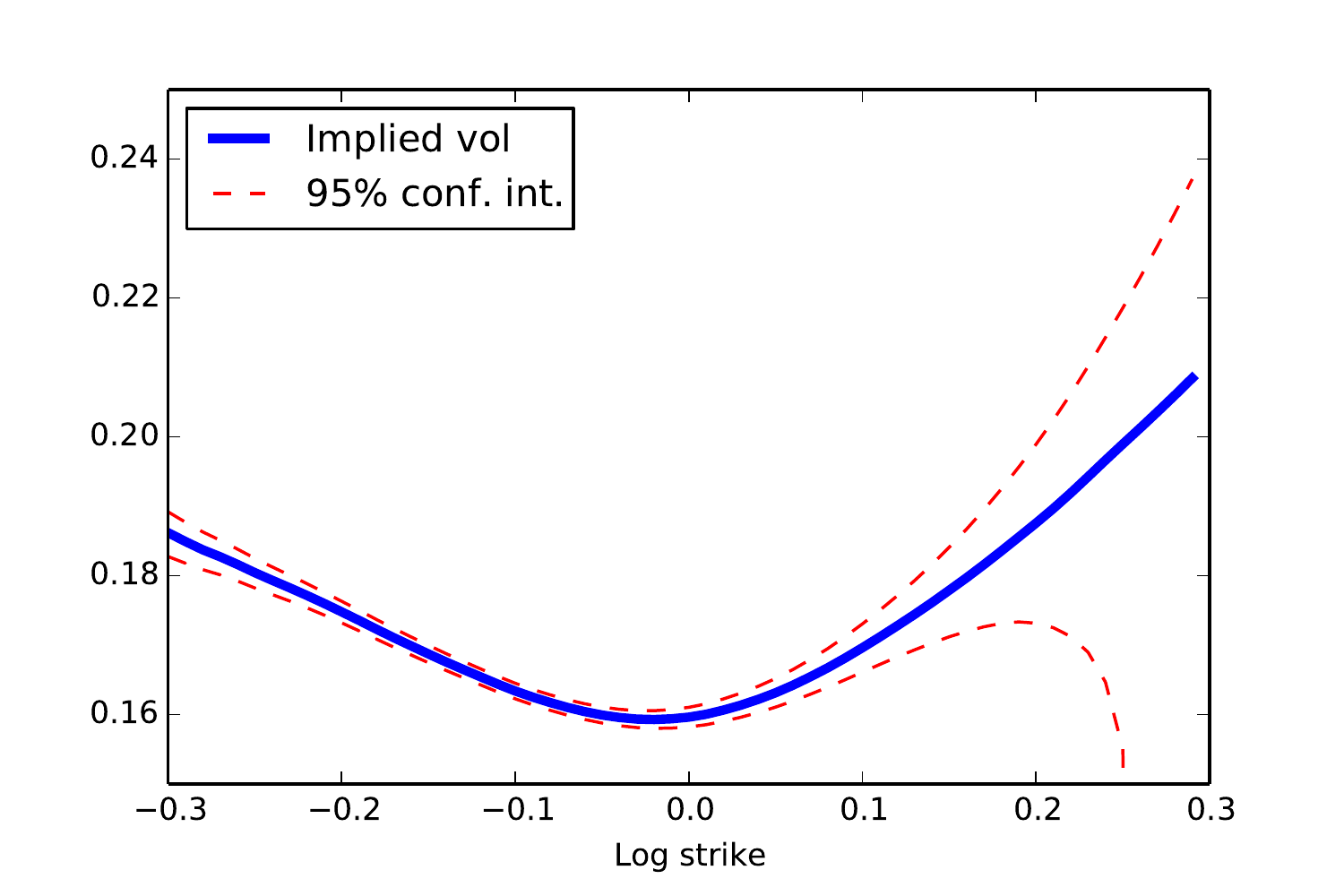}}
\caption{Basket implied volatility smile in the two-dimensional
  Wishart model. The upper and lower bounds correspond to the $95\%$
  confidence interval.}
\label{smile.fig}
\end{figure}

We next analyze the convergence of the renormalized implied volatility
smile to the long-maturity limit described in section
\ref{ivasymp.sec}. Figure \ref{ivasymp.fig}, shows the renormalized
smiles for different maturities together with the limiting
smile. These smiles were computed by Monte Carlo with 100,000
trajectories and a discretization time step $\Delta t = 0.1$. We
see that the convergence indeed appears to take place but it is quite
slow: even for 50-year maturity using the limit as the approximation
for the smile would lead to $10-15\%$ errors. 

\begin{figure}
\centerline{\includegraphics[width=0.7\textwidth]{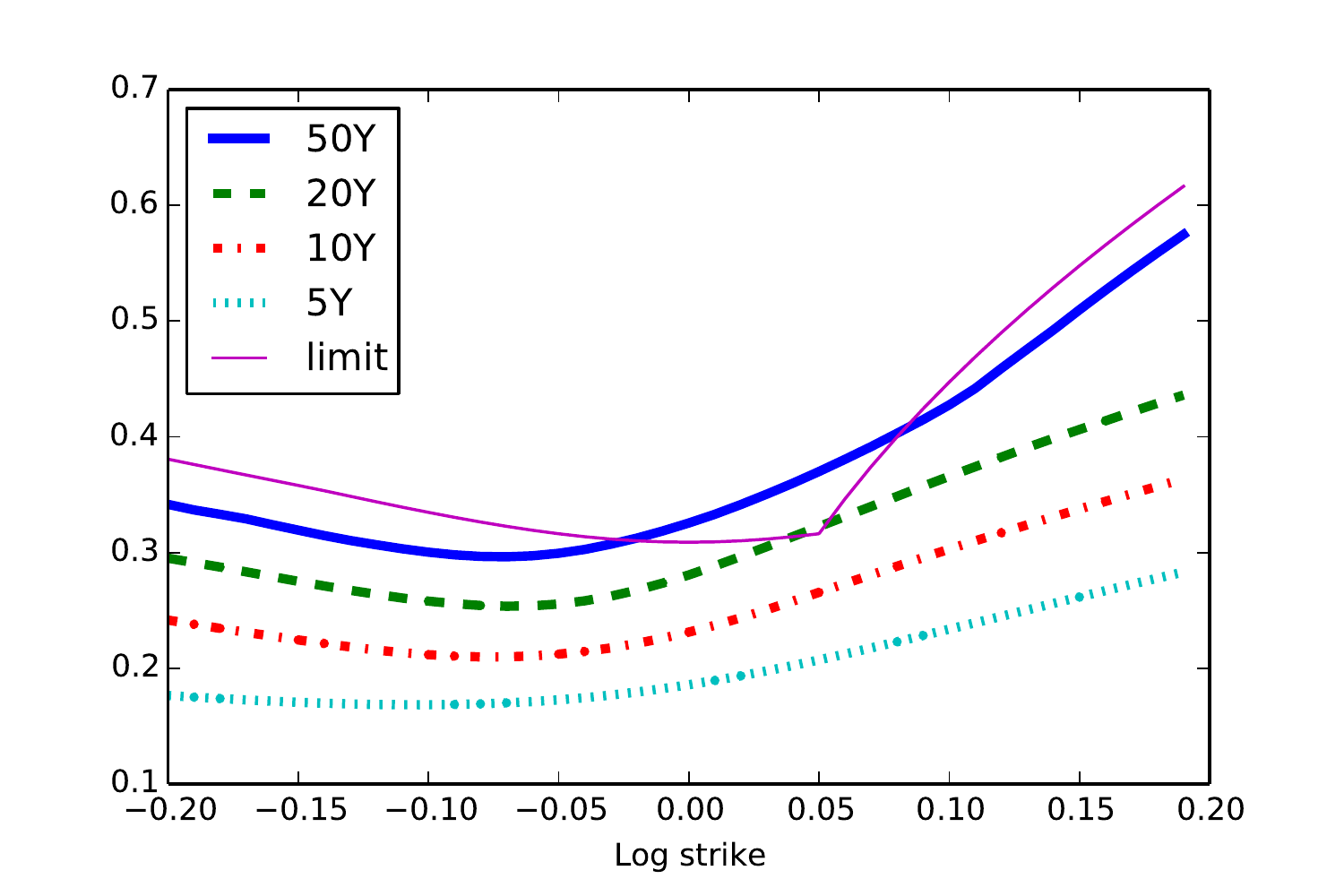}}
\caption{Convergence of the renormalized implied volatility
  smile to the theoretical limit in the Wishart model. }
\label{ivasymp.fig}
\end{figure}



\subsection{Variance reduction}

We now wish to test numerically the variance reduction method to price
basket put options. In order to do so, we first identify the law of
the Wishart process under the measure $\P_\theta$ and then calculate the asymptotically optimal measure change to finally test the method through Euler Monte-Carlo simulations. 

\subsubsection{Change of measure}

In order to simulate from the model under $\P_\theta$, we need the following result.
\begin{proposition}
Let $\theta \in \R^n$ be such that $\E[e^{\theta^\top Y_T}]<\infty$ and consider the change of measure $\frac{d\P_\theta}{d\P} = \frac{e^{\theta^\top Y_T}}{\E\left[e^{\theta^\top Y_T}\right]}$. Under $\P_\theta$, the process $(Y_t,X_t)$ has dynamics
\[
d Y_t = \left( r \mathbf{1} - \frac{1}{2} \left( (a^\top X_t \, a)_{11}\,,\,...\,,\,(a^\top X_t \, a)_{nn}\right)^\top + a^\top X_t \, a \, \theta \right) \, dt + a^\top X_t^{1/2} \, dZ_t^\theta
\]
and 
\begin{align*}
dX_t & = \left( \alpha I_n + (b + 2 \, \gamma_\theta(T-t)) X_t + X_t (b + 2 \, \gamma_\theta(T-t)) \right) \, dt \\
& \qquad \qquad \qquad \qquad \qquad \qquad + X_t^{1/2} \, dW_t^{\theta} + (dW_t^\theta)^\top X_t^{1/2} \:, \qquad X_0 = x \:,
\end{align*}
where $\gamma_\theta(t) = -\frac{1}{2}\left(V'(t,\theta) \,
  V^{-1}(t,\theta) + b\right)$, $V(t,\theta) = V(t)$ is given in
Proposition~\ref{prop:LaplaceExplicitForm} and $\left(Z_t^\theta\right)_{t \ge 0}$ and $\left(W_t^\theta\right)_{t \ge 0}$ are $\R^n$ and $\R^{n\times n}$-dimensional independent standard $\P_\theta$-Brownian motions.
\end{proposition}
\begin{proof}
By Equation \ref{eq:LaplaceEquation}, the Radon-Nikodym
density satisfies
\[
\zeta_t
:= \left.\frac{d\P_\theta}{d\P}\right|_{\F_t} \!\!\!\! = \frac{\E\!\left[e^{\theta^\top Y_T}\middle| \F_t \right]}{\E\left[e^{\theta^\top Y_T}\right]} = \frac{e^{\frac{\alpha}{2} \tr{b} t - \theta^\top \! Y_0 - r \theta^\top\! \mathbf{1} t - \tr{\gamma_\theta(T) \, x}} }{\det\!\left[ V(T,\theta)\right]^{\alpha/2}\!\det\!\left[ V(T-t,\theta)\right]^{-\alpha/2}}\, e^{\theta^\top \! Y_t + \tr{\gamma_\theta(T-t) X_t}} \:.
\]
By It\^o formula, the martingale property of $\zeta_t$, Equations \eqref{eq:X_t}
and \eqref{eq:Y_t}, and the properties of the trace, the dynamics of $\zeta_t$ is
\begin{align*}
d\zeta_t
& = \zeta_t \left( \theta^\top a^\top X_t^{1/2} \, dZ_t + \tr{\gamma_\theta(T-t) \, X_t^{1/2} \, dW_t} + \tr{\gamma_\theta(T-t) \, (dW_t)^\top \, X_t^{1/2}} \right) \\ 
& = \zeta_t \left( \theta^\top a^\top X_t^{1/2} \, dZ_t + 2\, \tr{\left(X_t^{1/2} \, \gamma_\theta(T-t) \right)^\top dW_t}\right) \:.
\end{align*}
Therefore, by Girsanov's theorem, 
\[
Z_t^\theta := Z_t - \int_0^t X_s^{1/2} a \,\theta \,ds
\]
and
\[
W_t^{\theta} := W_t - 2 \int_0^t  X_s^{1/2} \gamma_\theta(T-s) \, ds
\]
are $n$-dimensional and $n\times n$-dimensional standard $\P_\theta$-Brownian motions. Replacing $dZ_t$ and $dW_t$ in \eqref{eq:X_t} and \eqref{eq:Y_t} by their $\P_\theta$ versions finishes the proof.
\end{proof}

We note that $X$ is no longer a Wishart process under the probability~$\P_\theta$, since its dynamics has time-dependent coefficients. To sample paths on the time interval $[t_i,t_{i+1}]$, we use the exact scheme for the Wishart process with the coefficient $b+2\gamma_\theta(T-(t_i+t_{i+1})/2)$ instead of~$b$. As explained in~\cite{Alfonsibook} subsection~3.3.4 in the case of the CIR process with time-dependent coefficients, this leads to a second order scheme for the weak error. Then, we can approximate $Y$ in the same way as under $\P$:
\begin{align*}
  Y_{t_{i+1}} = Y_{t_i} &+ \left[r\mathbf 1-  \frac{1}{2}\text{diag}\left[ a^\top \frac{X_{t_i}+X_{t_{i+1}}}{2} a\right] + a^\top \frac{X_{t_i}+X_{t_{i+1}}}{2} a \theta \right] \Delta t\\& + \textup{Chol}\left(a^\top \frac{X_{t_i}+X_{t_{i+1}}}{2} a \right)(Z_{t_{i+1}} - Z_{t_i}),
\end{align*}
where $Z$ is a Brownian motion sampled independently from $X$. This gives a second order scheme for $(X,Y)$.

\subsubsection{Optimal variance reduction parameter for the European basket put option}
 
In this section, we compute the asymptotically optimal measure to
price basket put options with log-payoff $H(Y_T) = \log(K -
\omega^\top e^{Y_T})_+$, for some $\omega \in (\R_+^*)^n$. It is shown
in \cite[Section 4]{Ge:Ta} that the function $H$ is concave and that
its convex conjugate is given by
$$
\hat{H}(\theta) =  \left\{\begin{aligned}
&+\infty && \theta_k \geq 0 \text{ for some $k$}\\
& -\left(1-\sum_k \theta_k\right)\log \frac{1- \sum_k \theta_k}{K} - \sum_{k} \theta_k \log(-\theta_k/\omega_k) &&\text{otherwise.}
\end{aligned}\right.
$$
To compute the asymptotically optimal measure change parameter
$\theta^*$ using Theorem \ref{thm:OptimalTheta} we then minimize
$\hat{H}(\theta) + \Lambda(\theta)$ with a numerical convex
optimization algorithm. 

\subsubsection{Numerical simulations}

Let us now fix the parameters of the model to the values 
\[
b = -\left(\begin{array}{cc}
0.7 & 0.3 \\
0.3 & 0.5
\end{array}\right)\:, \qquad a = \left(\begin{array}{cc}
0.1 & 0 \\
0 & 0.12
\end{array}\right)\]
and $\alpha = 4.5$, with initial values $S_0= \mathds{1}$ and $x = \, I_2$ and consider the problem of pricing a basket put option with log-payoff 
\[
H(Y_T) = \log\left(K - \frac{1}{2}\, e^{Y_T^1} + \frac{1}{2}\, e^{Y_T^1}\right)_+.
\]
For a wide variety of maturities $T$ and strikes
$K$, listed in Table \ref{tab:BasketVarianceReduction}, we simulate
100,000 trajectories, using the discretization scheme described above,
with step size $\Delta = \frac{1}{40}$, under both measures $\P$ and
$\P_\theta$ for the asymptotically optimal $\theta$. The results are
presented in Table \ref{tab:BasketVarianceReduction}.   
\begin{table}[H]
\centering
\begin{tabular}{llllll}
\hline
Maturity, years & Strike & Price & Std. dev. & Var. ratio & Time, seconds\\
  \hline 
0.50	&0.7&	2.18e-07&	3.37e-08&	119&	202\\
0.50&	0.8&	3.29e-05&	9.5e-07&	22.5&	167\\
0.50&	0.9&	1.776e-03&	1.38e-05&	5.28&	169\\
0.50&	1.0&	2.6201e-02&	6.85e-05&	3.15	&167\\
0.50&	1.1&	1.0306e-01	&9.86e-05&	3.96&	167\\
0.50&	1.2&	2.0027e-01&	8.29e-05&	6.68&	167\\
0.50	&1.3&	3.0005e-01	&6.41e-05&	11.3&	180\\
0.50&	1.4&	3.9999e-01	&5.32e-05&	16.5&	168\\\hline
0.25&	1.0&	1.730e-02&	5.17e-05&	2.42&	92\\
1.00&	1.0&	4.115e-02&	9.51e-05&	3.76	&319\\
2.00&	1.0&	6.423e-02	&1.39e-04&	3.86&	618\\
3.00&	1.0&	8.319e-02	&1.78e-04&	3.63	&934\\
5.00	&1.0&	1.1579e-01&	2.46e-04	&3.22&	1522\\
\hline
\end{tabular}
\caption{The variance ratio as function of the maturity and the strike for the basket put option on the Wishart stochastic volatility model.}
\label{tab:BasketVarianceReduction}
\end{table}

The variance ratio is the ratio of the variance under the original
measure $\mathbb P$ to that under the asymptotically optimal measure
$\mathbb P_\theta$. As expected, the performance of the importance
sampling algorithm is best for options far from the money, when the
exercise is a rare event, but even for at the money options the
variance reduction factor is significant, of the order of 3--4. The
computational overhead for using the variance reduction algorithm is
small: it does not exceed $20\%$ for a small number of trajectories and
decreases with the number of trajectories because some precomputation
steps are performed only once.

{\bf Acknowledgements. } This research benefited from the support of the ``Chaire Risques Financiers'', Fondation du Risque. 


\end{document}